%% file: RatBaseRhythm_v9.tex
\newif\ifdraft  \draftfalse  
\newif\ifpdflatex  \pdflatexfalse 

\documentclass[12pt]{article}
\usepackage[
  a4paper
  ,margin=1.25in
  ,foot=0.5in
  ]{geometry}
\usepackage{xspace}
\usepackage{vaucanson-g}
\usepackage{pstricks-add}
\usepackage{scalefnt}
\usepackage{amssymb,amsmath}
\usepackage{amsthm}
\usepackage{amsfonts}
\usepackage{enumitem}\setlist{nolistsep}

\usepackage{mathtools}
\usepackage[mathlines]{lineno}
\usepackage[T1]{fontenc}
\input{accent_keys}
\usepackage{float}
\usepackage{MnSymbol}

\usepackage{algorithmic}
\algsetup{indent=2em,linenodelimiter=.}

\floatstyle{ruled}
\newfloat{procedure}{ht}{pc}
\floatname{procedure}{Procedure}
\usepackage{caption}
\usepackage{subfig}
\usepackage{graphicx}
\usepackage{multirow}
\title{Rhythmic generation 
       of  trees and languages}
\author
{%
        Victor Marsault\thanks{Corresponding author, victor.marsault@telecom-paristech.fr}\,
        \thanks{Telecom-ParisTech and CNRS, 46 rue Barrault, 75013 Paris, France}~
        \and 
        Jacques Sakarovitch\footnotemark[\value{footnote}]
}

\date{2014\,--\,05\,--\,05}

\input{\TEXINPUTSDIR vm_macros2}
\input{\TEXINPUTSDIR pretty}
\input{\TEXINPUTSDIR js_symboles3.tex}

\input{\TEXINPUTSDIR js_formating_macros3.tex}

\newcommandy[o]{\path}{}
\input{\TEXINPUTSDIR js_macros4.tex}

\newtheorem{innercustomthm}{\customthmname}
\newcommand{\customthmname}{Theorem}
\newenvironment{customthm}[2][Theorem]
  {%
  \renewcommand{\customthmname}{#1}%
  \renewcommand\theinnercustomthm{#2.}%
  \innercustomthm%
  }
  {\endinnercustomthm}
\renewcommand{\x}{\xmd\! \times \!\xmd}
\usepackage{stmaryrd}

\input{aux_macros}

\makeatletter

\let\c@table\c@figure
\makeatother 

\nolinenumbers

\ifdraft
  \linenumbers
  \usepackage[notref,notcite]{showkeys}
  \renewcommand*\showkeyslabelformat[1]%
      {\footnotesize\normalfont\ttfamily[#1]~~}
\fi

\SetVCDirectory{vcsg/}
\ifdraft\ShowFrame\fi
\newcommand{\scaleu}{1.4}
\newcommand{\scaled}{0.85}
\newcommand{\scalet}{0.3}
\newcommand{\scaleq}{0.38}
\newcommand{\alphxye}{{\{ x, y \}}^{*}}
\renewcommand{\pref}[1]{\Pre\,(#1)}
\newcommand{\App}{A_{p'}}                         
\newcommand{\Apps}{\App^{\xmd *}}                         
\newcommand{\Br}{B_{\rhth}}                         
\newcommand{\Brs}{\Br^{\xmd *}}                         
\newcommand{\rhth}{\rhythm} 
\newcommand{\rhthtp}{(r_0,r_1,\ldots,r_{\qmu})} 
\newcommand{\rhthpq}{\rhth_{\pqs}} 
\newcommand{\rhthct}{\rhth_{\cts}} 
\newcommand{\gr}{growth ratio\xspace}
\newcommand{\pthw}[1]{\mathtt{path}\hspace*{-0.1em}\left(#1\right)}
\newcommand{\sigs}{\signature[s]}
\newcommand{\lablbd}{\labelling[\lambda]}                         
\newcommand{\labgmm}{\labelling[\gamma]}                         
\newcommand{\labg}[1]{\labgmm_{#1}}                         
\newcommand{\labgp}{\labg{p}}                         
\newcommand{\labgpq}{\labg{\pqs}}                         
\newcommand{\labgrpq}{\labg{\rhthpq}}                         
\newcommand{\labgr}{\labg{\rhth}}                         
\newcommand{\labgptp}{(0,1,\ldots,\pmu)}                         
\newcommand{\labgmtp}{(\gamma_0,\gamma_1,\ldots,\gamma_{\pmu})}                         
\newcommand{\TreeSmb}{\Ic}
\newcommand{\Tree}[1]{\TreeSmb_{#1}}
\newcommand{\Trr}{\Tree{\rhth}}
\newcommand{\Trtuu}{\Tree{(3,1,1)}}
\newcommand{\Trdu}{\Tree{(2,1)}}

\newcommand{\Trtd}{\Tree{\tds}}
\newcommand{\KSmb}{K}
\newcommand{\Klg}[1]{\KSmb_{#1}}
\newcommand{\Klgr}{\Klg{\rhth}}
\newcommand{\Klgtuu}{\Klg{(3,1,1)}}

\newcommand{\prm}{\mathbf{\gamma}}
\newcommand{\prmpq}{\prm_{\frac{p}{q}}}
\newcommand{\chrw}[1]{\mathbf{w}_{#1}}
\newcommand{\wpq}{\chrw{\pqs}}

\newcommand{\wct}{\chrw{\cts}}
\begin{document}
\thispagestyle{plain}
\tolerance=5000


\maketitle

\begin{abstract}
  This work builds on the notion of \emph{breadth-first signature} of infinite 
    trees and (prefix-closed) languages introduced by the authors 
    in a previous work. 
  We focus here on periodic signatures, a case coming from 
    the study of rational base numeration systems; 
    the language of integer representations in base~$\frac{p}{q}$ has a purely periodic
    signature whose period is derived from the Christoffel word of slope~$\frac{p}{q}$.
  Conversely, we characterise languages whose signature are purely periodic as 
    representations of integers in such number systems with 
    non-canonical alphabets of digits.
\end{abstract}


\input{RBR-int}

\input{RBR-sec2}

\input{RBR-sec3}

\input{RBR-sec4}

\input{RBR-sec5}

\input{RBR-ccl}

%
\input{RBR-bib}
\end{document}

%% file: accent_keys.tex
\ifx\optionkeymacros\undefined\else \fi
\catcode`\=\active\def{\c{c}}      
\catcode`\'=\active\def'{\c{C}}      

\catcode`\Ï=\active\defÏ{{\oe}}      
\catcode`\Î=\active\defÎ{{\OE}}      

\catcode`\¾=\active\def¾{{\ae}}      
\catcode`\®=\active\def®{{\AE}}      

\catcode`\Š=\active\defŠ{\"a}        
\catcode`\‰=\active\def‰{\^a}        
\catcode`\ˆ=\active\defˆ{\`a}        
\catcode`\‡=\active\def‡{\'a}        
\catcode`\‹=\active\def‹{\~a}        
\catcode`\€=\active\def€{\"A}        
\catcode`\å=\active\defå{\^A}        
\catcode`\Ë=\active\defË{\`A}        
\catcode`\ç=\active\defç{\'A}        
\catcode`\Ì=\active\defÌ{\~A}        

\catcode`\'=\active\def'{\"e}        
\catcode`\=\active\def{\^e}        
\catcode`\=\active\def{\`e}        
\catcode`\Ž=\active\defŽ{\'e}        
\catcode`\è=\active\defè{\"E}        
\catcode`\æ=\active\defæ{\^E}        
\catcode`\é=\active\defé{\`E}        
\catcode`\ƒ=\active\defƒ{\'E}        

\catcode`\•=\active\def•{\"{\i}}     
\catcode`\"=\active\def"{\^{\i}}     
\catcode`\"=\active\def"{\`{\i}}     
\catcode`\'=\active\def'{\'{\i}}     
\catcode`\ì=\active\defì{\"I}        
\catcode`\ë=\active\defë{\^I}        
\catcode`\í=\active\defí{\`I}        
\catcode`\ê=\active\defê{\'I}        

\catcode`\š=\active\defš{\"o}        
\catcode`\™=\active\def™{\^o}        
\catcode`\˜=\active\def˜{\`o}        
\catcode`\—=\active\def—{\'o}        
\catcode`\›=\active\def›{\~o}        
\catcode`\…=\active\def…{\"O}        
\catcode`\ï=\active\defï{\^O}        
\catcode`\ñ=\active\defñ{\`O}        
\catcode`\î=\active\defî{\'O}        
\catcode`\Í=\active\defÍ{\~O}        

\catcode`\Ÿ=\active\defŸ{\"u}        
\catcode`\ž=\active\defž{\^u}        
\catcode`\=\active\def{\`u}        
\catcode`\œ=\active\defœ{\'u}        
\catcode`\†=\active\def†{\"U}        
\catcode`\ó=\active\defó{\^U}        
\catcode`\ô=\active\defô{\`U}        
\catcode`\ò=\active\defò{\'U}        

\catcode`\Ø=\active\defØ{\"y}        
\catcode`\Ù=\active\defÙ{\"Y}        

\catcode`\–=\active\def–{\~n}        
\catcode`\"=\active\def"{\~N}        

\let\optionkeymacros\null

%% file: vm_macros2.tex
\usepackage{ifthen}
\usepackage{etoolbox}
\usepackage{xargs}
\usepackage{slantsc}

\newcommandx{\newtheoremy}[3][2={}]{
  \ifthenelse{\equal{#2}{}}{
    \ifcsmacro{#1}{}{\newtheorem{#1}{#3}}
  }{
    \ifcsmacro{#1}{}{\newtheorem{#1}[#2]{#3}}
  }
}

\newcommand{\thmBlockFont}[1]{#1}
\newtheoremy{thm}{thm}
\newtheoremy{algorithm}[thm]{\thmBlockFont{Algorithm}}
\newtheoremy{corollary}[thm]{\thmBlockFont{Corollary}}
\newtheoremy{conjecture}[thm]{\thmBlockFont{Conjecture}}
\newtheoremy{definition}[thm]{\thmBlockFont{Definition}}
\newtheoremy{example}[thm]{\thmBlockFont{Example}}
\newtheoremy{lemma}[thm]{\thmBlockFont{Lemma}}
\newtheoremy{proposition}[thm]{\thmBlockFont{Proposition}}
\newtheoremy{property}[thm]{\thmBlockFont{Property}}
\newtheoremy{question}[thm]{\thmBlockFont{Question}}
\newtheoremy{remark}[thm]{\thmBlockFont{Remark}}
\newtheoremy{notation}[thm]{\thmBlockFont{Notation}}
\newtheoremy{theorem}[thm]{\thmBlockFont{Theorem}}

\newcommand{\ldefinition}[1]{\label{d.#1}}
\newcommand{\lexample}[1]{\label{e.#1}}
\newcommand{\llemma}[1]{\label{l.#1}}
\newcommand{\lproposition}[1]{\label{p.#1}}

\newcommand{\lsection}[1]{\label{s.#1}}

\newcommand{\lfigure}[1]{\label{f.#1}}
\newcommand{\ltheorem}[1]{\label{t.#1}}
\newcommand{\lequation}[1]{\label{eq.#1}}

\newcommand{\generalref}[2]{%
  \ifthenelse{\equal{#1}{eq}}%
  {(\ref{#1.#2})}%
  {\ref{#1.#2}}%
}
\newcommand{\generalpageref}[2]{\pageref{#1.#2}}

\makeatletter
\newcommand*{\ralgorithm}{\@ifstar{\generalref{a}}{Algorithm~\ralgorithm*}}
\newcommand*{\palgorithm}{\@ifstar{\generalpageref{a}}{page~\palgorithm*}}

\newcommand*{\rcorollary}{\@ifstar{\generalref{c}}{Corollary~\rcorollary*}}
\newcommand*{\pcorollary}{\@ifstar{\generalpageref{c}}{page~\pcorollary*}}

\newcommand*{\rconjecture}{\@ifstar{\generalref{cj}}{Conjecture~\rconjecture*}}
\newcommand*{\pconjecture}{\@ifstar{\generalpageref{cj}}{page~\pconjecture*}}

\newcommand*{\rdefinition}{\@ifstar{\generalref{d}}{Definition~\rdefinition*}}
\newcommand*{\pdefinition}{\@ifstar{\generalpageref{d}}{page~\pdefinition*}}

\newcommand*{\rexample}{\@ifstar{\generalref{e}}{Example~\rexample*}}
\newcommand*{\pexample}{\@ifstar{\generalpageref{e}}{page~\pexample*}}

\newcommand*{\rlemma}{\@ifstar{\generalref{l}}{Lemma~\rlemma*}}
\newcommand*{\plemma}{\@ifstar{\generalpageref{l}}{page~\plemma*}}

\newcommand*{\rproposition}{\@ifstar{\generalref{p}}{Proposition~\rproposition*}}
\newcommand*{\pproposition}{\@ifstar{\generalpageref{p}}{page~\pproposition*}}

\newcommand*{\rproperty}{\@ifstar{\generalref{pp}}{Proposition~\rproperty*}}
\newcommand*{\pproperty}{\@ifstar{\generalpageref{pp}}{page~\pproperty*}}

\newcommand*{\rprocedure}{\@ifstar{\generalref{pc}}{Procedure~\rprocedure*}}
\newcommand*{\pprocedure}{\@ifstar{\generalpageref{pc}}{page~\pprocedure*}}

\newcommand*{\rremark}{\@ifstar{\generalref{r}}{Remark~\rremark*}}
\newcommand*{\premark}{\@ifstar{\generalpageref{r}}{page~\premark*}}

\newcommand*{\rnotation}{\@ifstar{\generalref{n}}{Notation~\rnotation*}}
\newcommand*{\pnotation}{\@ifstar{\generalpageref{n}}{page~\pnotation*}}

\newcommand*{\rsection}{\@ifstar{\generalref{s}}{Section~\rsection*}}
\newcommand*{\psection}{\@ifstar{\generalpageref{s}}{page~\psection*}}

\newcommand*{\rtable}{\@ifstar{\generalref{t}}{Table~\rtable*}}
\newcommand*{\ptable}{\@ifstar{\generalpageref{t}}{page~\ptable*}}

\newcommand*{\rfigure}{\@ifstar{\generalref{f}}{Figure~\rfigure*}}
\newcommand*{\pfigure}{\@ifstar{\generalpageref{f}}{page~\pfigure*}}

\newcommand*{\requation}{\@ifstar{\generalref{eq}}{Equation~\requation*}}
\newcommand*{\pequation}{\@ifstar{\generalpageref{eq}}{page~\pequation*}}

\newcommand*{\rtheorem}{\@ifstar{\generalref{t}}{Theorem~\rtheorem*}}
\newcommand*{\ptheorem}{\@ifstar{\generalpageref{t}}{page~\ptheorem*}}

\makeatother

\usepackage{atbegshi,picture}
\def\Vhrulefill{\leavevmode\leaders\hrule height 0.7ex depth \dimexpr0.4pt-0.7ex\hfill\kern0pt}

\newcommandx{\wlen}[1]{|#1|}

\newcommandx{\cod}[1]{\langle #1 \rangle}
\newcommandx{\floor}[1]{\lfloor #1 \rfloor}
\newcommandx{\bfloor}[1]{\left\lfloor #1 \right\rfloor}
\newcommandx{\bceil}[1]{\left\lceil #1 \right\rceil}
\newcommandx{\ceil}[1]{\lceil #1 \rceil}

\newcommandx{\newcommandy}[5][1=i,3=0,4={}]{%
  \ifthenelse{\isundefined{#2}}{\newcommandx{#2}[#3][#4]{#5}}{%
      \ifthenelse{\equal{#1}{i}}{}{}%
      \ifthenelse{\equal{#1}{o}}{\renewcommandx{#2}[#3][#4]{#5}}{}%
    }%
}
\newcommandy[o]{\pathx}[2][2={}]{
  \nlb%
  \ifthenelse{\equal{#2}{}}%
    {\xrightarrow{\ #1 \ }}%
    {\underset{#2}{\xrightarrow{\ #1 \ }}}%
  \nlb%
}
\newcommandy[o]{\pathy}[2][2={}]{\nlb\underset{#2}{\xleftarrow{\ #1 \ }}\nlb}

\newcommand{\val}[1]{\widebar{#1}}

\newcommand{\strong}[1]{\textbf{#1}}
\newcommand{\ssc}[1]{\textbf{\textsc{#1}}}
\newcommand{\pref}[1]{\text{Pref}\,(#1)}

\newcommand{\set}[1]{\{#1\}}

\newcommand{\Z}{\mathbb{Z}}
\newcommand{\N}{\mathbb{N}}
\newcommand{\Q}{\mathbb{Q}}

\newcommand{\widebar}{\overline}

\newcommandy[o]{\Cup}{\bigcup}
\newcommand{\nlb}{\nolinebreak}

\renewcommand{\thmBlockFont}[1]{\ssc{#1}}

\newcommand{\vmfiguretodo}[2][]{%
  \begin{figure}[ht!]
    \frame{%
      \begin{minipage}{\linewidth}
        ~\hfill~
        \vspace*{#2cm}
      \end{minipage}
    }
    \ifthenelse{\equal{#1}{}}{}{\caption{#1}}
  \end{figure}
}

\makeatletter
\newcommandx{\newcommandWithStar}[3][1=i]{%
  \newcommandy[#1]{#2}{\protect\@ifstar{\leavevmode\protect\nlb$\protect#3$}{#3}}
}
\makeatother

%% file: pretty.tex
\renewcommand{\leq}{\leqslant}
\renewcommand{\geq}{\geqslant}
\renewcommand{\phi}{\varphi}
\renewcommand{\epsilon}{\varepsilon}
\renewcommand{\mod}{\text{~mod~}}

%% file: js_symboles3.tex
%
%
%
%
\newcommand{\fa}{\forall}
\newcommand{\ext}{\exists}



\newcommand{\bk}{\mathrel{\backslash }}

\newcommand{\e}{\text{\quad}}                 
\newcommand{\ee}{\text{\qquad}}               
\newsavebox{\InterSymbolSpace}
\savebox{\InterSymbolSpace}{\hspace{0.125em}}
\newsavebox{\SideFormulaSpace}
\savebox{\SideFormulaSpace}{\hspace{0.2em}}
\newcommand{\msp}{\usebox{\SideFormulaSpace}} 
\newcommand{\xmd}{\usebox{\InterSymbolSpace}} 
\newcommand{\eqpnt}{\makebox[0pt][l]{\: .}}

\newcommand{\eqvrg}{\makebox[0pt][l]{\: ,}}
\newcommand{\EqVrgInt}{\: , \e }

\newcommand{\quantvrg}{\, , \;}
\newcommand{\quantsp}{\ee }
\newcommand{\quantsmsp}{\e }
%


%

%

\newcommand{\LatinLocution}[1]{{\itshape #1}\xspace}

\newcommand{\cf}{\LatinLocution{cf.}}




%

%



%
\newcommand{\UNmbb}{{\mathchoice
{\hbox{$\textstyle\rm 1\kern-0.2em I$}}%
{\hbox{$\textstyle\rm 1\kern-0.2em I$}}%
{\hbox{$\scriptstyle\rm 1\kern-0.15em I$}}%
{\hbox{$\scriptscriptstyle\rm 1\kern-0.1em I$}}%
}}

\newcommand{\Ic}{\mathcal{I}}

\newcommand{\Tc}{\mathcal{T}}

%

%


%


%


%










%
\newlength{\ArrowDiagSize}
\setlength{\ArrowDiagSize}{6pt}
\newlength{\ArrowDiagWidth}
\setlength{\ArrowDiagWidth}{2pt}
\newpsstyle{SLDiagStyle}%
   {colsep=6ex,rowsep=5ex,nodesep=1ex,npos=.45,%
    arrows=->,linewidth=\ArrowDiagWidth,arrowsize=\ArrowDiagSize,%
        linestyle=solid,linecolor=\ArrowDiagColor}

\newenvironment{SLDiag}%
   {\psset{style=SLDiagStyle}\begin{psmatrix}}%
   {\end{psmatrix}}%
\newcommand{\CDSL}{\begin{SLDiag}}
\newcommand{\CDSLF}{\end{SLDiag}}
\newenvironment{DiagraBig}%
{\psmatrix[colsep=7ex,rowsep=6ex,arrows=->,nodesep=1ex,npos=.45]}%
{\endpsmatrix}
\newcommand{\CDB}{\begin{DiagraBig}}
\newcommand{\CDBF}{\end{DiagraBig}}
\newenvironment{DiagraSmall}%
{\psmatrix[colsep=3ex,rowsep=3ex,arrows=->,nodesep=1ex,npos=.45]}%
{\endpsmatrix}
\newcommand{\CDS}{\begin{DiagraSmall}}
\newcommand{\CDSF}{\end{DiagraSmall}}


%
%
\newcommand{\matriceuu}[1]%
    {\begin{pmatrix} #1 \end{pmatrix}}
\newcommand{\matricedd}[4]%
    {\begin{pmatrix} #1 & #2 \\ #3 & #4 \end{pmatrix}}
\newcommand{\vecteurd}[2]%
    {\begin{pmatrix} #1 \\ #2 \end{pmatrix}}
\newcommand{\ligned}[2]%
    {\begin{pmatrix} #1 & #2 \end{pmatrix}}
\newcommand{\matricett}[9]%
    {\begin{pmatrix}  #1 & #2 & #3 \\
                      #4 & #5 & #6 \\
                      #7 & #8 & #9 \end{pmatrix}}
\newcommand{\vecteurt}[3]%
    {\begin{pmatrix} #1 \\ #2 \\ #3 \end{pmatrix}}
\newcommand{\lignet}[3]%
    {\begin{pmatrix} #1 & #2 & #3 \end{pmatrix}}

\newlength{\jsWidthCol}
\setlength{\jsWidthCol}{0pt}

\newlength{\blocinterligne}
\setlength{\blocinterligne}{1.4ex}

\newlength{\blocinterligned}
\setlength{\blocinterligned}{2ex}

%

%
%
%
\newlength{\temparraycolsep}
\newlength{\longueurbloc}
\newlength{\hauteurbloc}
\newlength{\centragebloc}
\setlength{\longueurbloc}{9ex}
\setlength{\hauteurbloc}{7ex}
\setlength{\centragebloc}{-3ex}
\newlength{\longueurblc}
\newlength{\hauteurblc}
\newlength{\centrageblc}
\setlength{\longueurblc}{6.5ex}
\setlength{\hauteurblc}{5ex}
\setlength{\centrageblc}{-2ex}
\newcommand{\blocligne}[1]%
    {\framebox[\longueurbloc]{$#1$}}
\newcommand{\blocmatrice}[1]%
    {\framebox[\longueurbloc]{\rule[\centragebloc]{0mm}{\hauteurbloc}$#1$}}
\newcommand{\blocvecteur}[1]%
    {\framebox{\rule[\centragebloc]{0mm}{\hauteurbloc}$#1$}}
\newcommand{\blcligne}[1]%
    {\framebox[\longueurblc]{$#1$}}
\newcommand{\blcmatrice}[1]%
    {\framebox[\longueurblc]{\rule[\centrageblc]{0mm}{\hauteurblc}$#1$}}
\newcommand{\blcvecteur}[1]%
    {\framebox{\rule[\centrageblc]{0mm}{\hauteurblc}$#1$}}
%
\newcommand{\matriceddblvs}[4]
   {\setlength{\temparraycolsep}{\arraycolsep}%
    \setlength{\arraycolsep}{1.3pt}%
        \left (%
    \begin{array}{cc}%
                #1  & \blcligne{#2} \\
            \blcvecteur{#3} & \blcmatrice{#4}
        \end{array}%
        \right )%
    \setlength{\arraycolsep}{\temparraycolsep}%
   }%
\newcommand{\vecteurdblvs}[2]%
   {\setlength{\temparraycolsep}{\arraycolsep}%
    \setlength{\arraycolsep}{1.5pt}%
        \left (%
    \begin{array}{c}%
                #1  \\
            \blcvecteur{#2}
        \end{array}%
        \right )%
    \setlength{\arraycolsep}{\temparraycolsep}%
   }%
\newcommand{\lignedblvs}[2]%
   {\setlength{\temparraycolsep}{\arraycolsep}%
    \setlength{\arraycolsep}{1.5pt}%
        \left (%
    \begin{array}{cc}%
                #1  & \blcligne{#2}
        \end{array}%
        \right )%
    \setlength{\arraycolsep}{\temparraycolsep}%
   }%
%
\newcommand{\matricettblvs}[9]
   {\setlength{\temparraycolsep}{\arraycolsep}%
    \setlength{\arraycolsep}{1.5pt}%
        \left (%
    \begin{array}{ccc}%
                #1  & \blcligne{#2} & #3\\
            \blcvecteur{#4} & \blcmatrice{#5} & \blcvecteur{#6}\\
                #7  & \blcligne{#8} & #9\\
        \end{array}%
        \right )%
    \setlength{\arraycolsep}{\temparraycolsep}%
   }%
\newcommand{\vecteurtblvs}[3]%
   {\setlength{\temparraycolsep}{\arraycolsep}%
    \setlength{\arraycolsep}{1.5pt}%
        \left (%
    \begin{array}{c}%
                #1  \\
            \blcvecteur{#2}\\
                #3
        \end{array}%
        \right )%
    \setlength{\arraycolsep}{\temparraycolsep}%
   }%
\newcommand{\lignetblvs}[3]%
   {\setlength{\temparraycolsep}{\arraycolsep}%
    \setlength{\arraycolsep}{1.5pt}%
        \left (%
    \begin{array}{ccc}%
                #1  & \blcligne{#2} & #3
        \end{array}%
        \right )%
    \setlength{\arraycolsep}{\temparraycolsep}%
   }%
%
\newcommand{\matricettblblvs}[9]
   {\setlength{\temparraycolsep}{\arraycolsep}%
    \setlength{\arraycolsep}{1.5pt}%
        \left (%
    \begin{array}{ccc}%
                #1  & \blcligne{#2} & \blcligne{#3}\\
            \blcvecteur{#4} & \blcmatrice{#5} & \blcmatrice{#6}\\
                \blcvecteur{#7}  & \blcmatrice{#8} & \blcmatrice{#9}\\
        \end{array}%
        \right )%
    \setlength{\arraycolsep}{\temparraycolsep}%
   }%
\newcommand{\vecteurtblblvs}[3]%
   {\setlength{\temparraycolsep}{\arraycolsep}%
    \setlength{\arraycolsep}{1.5pt}%
        \left (%
    \begin{array}{c}%
                #1  \\
            \blcvecteur{#2}\\
                \blcvecteur{#3}
        \end{array}%
        \right )%
    \setlength{\arraycolsep}{\temparraycolsep}%
   }%
\newcommand{\lignetblblvs}[3]%
   {\setlength{\temparraycolsep}{\arraycolsep}%
    \setlength{\arraycolsep}{1.5pt}%
        \left (%
    \begin{array}{ccc}%
                #1  & \blcligne{#2} & \blcligne{#3}
        \end{array}%
        \right )%
    \setlength{\arraycolsep}{\temparraycolsep}%
   }%
%

%% file: js_formating_macros3.tex
\newcommand{\medskipneg}{\vspace*{-2ex}} 
\newcommand{\smallskipneg}{\vspace*{-1ex}} 
%
 %

\newlength{\DefiTest}\setlength{\DefiTest}{0pt}%
\newlength{\DefiHeightu}\newlength{\DefiHeightd}%
\newlength{\DefiDepthu}\newlength{\DefiDepthd}%
\newcommand{\Defi}[2]%
    {%
     \settoheight{\DefiHeightu}{${\displaystyle #1}$}%
     \settodepth{\DefiDepthu}{${\displaystyle #1}$}%
     \addtolength{\DefiHeightu}{\DefiDepthu}%
     \settoheight{\DefiHeightd}{${\displaystyle #2}$}%
     \settodepth{\DefiDepthd}{${\displaystyle #2}$}%
     \addtolength{\DefiHeightd}{\DefiDepthd}%
     \left\{#1%
     \rule[-\DefiDepthd]{\DefiTest}{\DefiHeightd}%
     \xmd\right|%
     \left.%
     \rule[-\DefiDepthu]{\DefiTest}{\DefiHeightu}%
      #2\right\}%
     }
\newlength{\ColoText}
\newlength{\ColoFigu}
\newlength{\TextFiguSpace}
\newlength{\parindenttemp} 
\newlength{\parskiptemp} 
\newlength{\fboxseptemp} 
\newcommand{\TFBoxing}{}
\newcommand{\TFVertAlig}{}
\newcommand{\LeftLarg}{}
\setlength{\fboxseptemp}{\fboxsep}
\setlength{\parindenttemp}{\parindent}
\setlength{\parskiptemp}{\parskip}
\setlength{\TextFiguSpace}{1.2em}
\renewcommand{\LeftLarg}{.66}
\ifdraft\renewcommand{\TFBoxing}{\fbox}\fi
\newcommand{\TxtFg}[3]%
   {%
    \setlength{\ColoText}{#1\textwidth}%
    \setlength{\ColoFigu}{\textwidth}%
    \addtolength{\ColoFigu}{-\ColoText}%
    \addtolength{\ColoText}{-.5\TextFiguSpace}%
    \addtolength{\ColoFigu}{-.5\TextFiguSpace}%
    \ifdraft\setlength{\fboxsep}{0pt}\fi
    \noi
    \TFBoxing{%
       \begin{minipage}[\TFVertAlig]{\ColoText}%
          \setlength{\parindent}{\parindenttemp}%
          \setlength{\parskip}{\parskiptemp}%
          \par\vspace*{0mm}
             #2
       \end{minipage}%
             }%
    \hspace*{\TextFiguSpace}%
    \TFBoxing{%
       \begin{minipage}[\TFVertAlig]{\ColoFigu}%
          \par\vspace*{0mm}%
             #3%
       \end{minipage}%
             }%
    \ifdraft\setlength{\fboxsep}{\fboxseptemp}\fi
   }%
\newcommand{\TextFigu}[3][\LeftLarg]%
   {\renewcommand{\TFVertAlig}{t}\TxtFg{#1}{#2}{#3}}
\newcommand{\TextFiguC}[3][\LeftLarg]%
   {\renewcommand{\TFVertAlig}{c}\TxtFg{#1}{#2}{#3}}
\newcommand{\TextFiguX}[3][\LeftLarg]
   {%
    \setlength{\ColoText}{#1\textwidth}%
    \setlength{\ColoFigu}{\textwidth}%
    \addtolength{\ColoFigu}{-\ColoText}%
    \addtolength{\ColoText}{-.5\TextFiguSpace}%
    \addtolength{\ColoFigu}{-.5\TextFiguSpace}%
    \addtolength{\ColoFigu}{\ETAExtendedLineWidth}
    \ifdraft\setlength{\fboxsep}{0pt}\fi
    \noi
    \ifodd\value{page}%
       \TFBoxing{%
          \begin{minipage}[t]{\ColoText}%
             \RstBLS
             \setlength{\parindent}{\parindenttemp}%
             \setlength{\parskip}{\parskiptemp}%
             \par\vspace*{0mm}
                #2
          \end{minipage}%
                }%
       \hspace*{\TextFiguSpace}%
       \TFBoxing{%
          \begin{minipage}[t]{\ColoFigu}%
             \par\vspace*{0mm}%
                #3%
          \end{minipage}%
                }%
    \else
       \hspace*{-\ETAExtendedLineWidth}
       \TFBoxing{%
          \begin{minipage}[t]{\ColoFigu}%
             \par\vspace*{0mm}%
                #3%
          \end{minipage}%
                }%
       \hspace*{\TextFiguSpace}%
       \TFBoxing{%
          \begin{minipage}[t]{\ColoText}%
             \RstBLS
             \setlength{\parindent}{\parindenttemp}%
             \setlength{\parskip}{\parskiptemp}%
             \par\vspace*{0mm}
                #2
          \end{minipage}%
                }%
    \fi%
    \ifdraft\setlength{\fboxsep}{\fboxseptemp}\fi
   }

\newcommand{\NoteEnMarge}[1]%
   {%
    \marginpar[\begin{flushright}%
               {\sl {\scriptsize #1}}%
               \end{flushright}]%
              {\begin{flushleft}%
               {\sl {\scriptsize #1}}%
               \end{flushleft}}%
	}%
\newcommand{\Axio}[1]%
   {\pointn #1\hspace*{.1em}\jspointtiret\hspace*{.4em}\ignorespaces}

%% file: js_macros4.tex
\newcommand{\ExtnF}[1]%
   {\overset{{\scriptscriptstyle \pmb{\smile}}}{#1}}



\newcommand{\DiffF}[1]%
   {\overset{{\scriptscriptstyle \pmb{\lor}}}{#1}}

\newcommand{\LocaF}[1]%
   {\overset{{\scriptscriptstyle \leftrightarrow}}{#1}}

%



\newcommand{\jsDist}[2][{}]%
   {\operatorname{\mathbf{d}_{#1}}\left(#2\right)}






\newcommand{\Pre}{{\operatorname{{\mathsf{Pre}}}}}



%

%

%

%


\renewcommand{\lim}{{\operatornamewithlimits{\mathsf{lim}}}}







%

%



%















%

\newcommand{\x}{\! \times \!}

%


%

%







\newcommand{\SerSAnMon}[2]%
    {#1 \langle \! \langle  #2  \rangle \! \rangle }
\newcommand{\SerSAnMonD}[2]%
    {\left[#1\right] \langle \! \langle  #2  \rangle \! \rangle }
\newcommand{\SerMon}[1]%
    {\!\langle \! \langle  #1  \rangle \! \rangle }
%


%



%


%
\newcommand{\PolSAnMon}[2]%
    {{#1 \langle  #2 \rangle }}
\newcommand{\PolMon}[1]%
    {{\!\langle  #1 \rangle }}

%


%

%
%

%




%

%
%

%


%


%
%



%

%

%

%



%

%
\newcommand{\jsStar}[1]{{{#1}^{*}}}
\newcommand{\Ae}{\jsStar{A}}
\newcommand{\Be}{\jsStar{B}}

\newcommand{\jsPlus}[1]{{{#1}^{+}}}
\newcommand{\Ap}{\jsPlus{A}}






%


%



%


%




%



\newcommand{\iotaK}{\iota_{\ShiftInd{K}}}

%



%

%

%


%



%

%

%

\newcommand{\compos}{\ccdot }











%

%







%

\newcommand{\phiikpsi}%
{{\varphi ^{-1}\! \compos        \iotaK \! \compos \! \psi }}
\newcommand{\phiiotpsi}[1]%
{{\varphi ^{-1}\! \compos        \iota _{\ShiftInd{#1}} \! \compos \! \psi }}
\newcommand{\phiintkpsi}[1]%
{{(#1\varphi ^{-1}\! \cap K) \psi }}

%

\newcommand{\jsgeq}{\geqslant }
\newcommand{\jsless}
   {\mathrel{\leqslant_{\!\!\!\!\scriptscriptstyle{/}}}}
\newcommand{\jsgrea}
   {\mathrel{\geqslant_{\!\!\!\!\scriptscriptstyle{\backslash}}}}

\newcommand{\lexiconeq}
   {\preccurlyeq_{\!\!\!\!\!\scalebox{1.8 1}{\scriptscriptstyle{\pmb{/}}}}}



%
\newcommand{\jsAutUn}[1]%
   {\mbox{$\left\langle \thinspace #1 \thinspace \right\rangle $}}












%
%



%

%


%

%

\newcommand{\ShiftInd}[1]{\raisebox{-0.3ex}{$\scriptstyle{#1}$}}

%

%

%

%

%


\newcommand{\actb}{\mathbin{\raisebox{0.2ex}%
                        {${\scriptscriptstyle \circ} $}}}
\newcommand{\ccdot}{\actb} 




%

%

%

\newlength{\vbh}\newlength{\vbd}\newlength{\vbt}%
\newcommand{\CompAuto}[1]%
    {%
     \settodepth{\vbd}{\mbox{$\displaystyle{#1\strut}$}}%
     \settoheight{\vbh}{\mbox{$\displaystyle{#1\strut}$}}%
     \setlength{\vbt}{\vbh}\addtolength{\vbt}{\vbd}%
     {}%
     \psline[linewidth=0.8pt]{c-c}(0,-.65\vbd)(0,.9\vbh)%
     \hspace*{0.7pt}%
     {#1}%
     \kern0.8pt%
     \psline[linewidth=0.8pt]{c-c}(0,-.65\vbd)(0,.9\vbh)%
     }%




\newcommand{\bornedeuxlignes}[2]%
{\mbox{$
\begin{array}{c}{\scriptstyle #1}\\ {\scriptstyle #2} \end{array}
       $}}

%


\renewcommand{\path}[1]{\xrightarrow{\ #1 \ }} 
\newcommand{\pathaut}[2]{\underset{#2}{\path{#1}}}






%
\newcommand{\ExpDer}[2][a]%
    {\operatorname{\frac{\partial}{\partial \mbox{$#1$}}}#2}
\newcommand{\ExpDerP}[2][a]%
    {\operatorname{\frac{\partial}{\partial\mbox{$#1$}}}\left(#2\right)}
%
%
\newcommand{\ExpDerr}[2][a]%
    {\operatorname{\frac{\partial_{\mathrm{R}}}{\partial \mbox{$#1$}}}#2}
\newcommand{\ExpDerB}[2][a]%
   {\operatorname{\frac{\partial_\mathsf{b}}{\partial \mbox{$#1$}}}#2}
\newcommand{\ExpDerBP}[2][a]%
   {\operatorname{\frac{\partial_\mathsf{b}}{\partial \mbox{$#1$}}}\left(#2\right)}

%

%


%

%

%


%




%% file: aux_macros.tex
\usepackage{scalefnt}

\newcommand{\defin}[1]{Definition~\ref{d.#1}}

\newcommand{\figur}[1]{Fig.\xmd\ref{f.#1}}
\newcommand{\lemme}[1]{Lemma~\ref{l.#1}}
\newcommand{\propo}[1]{Proposition~\ref{p.#1}}

\newcommand{\theor}[1]{Theorem~\ref{t.#1}}
\newlength{\retraita}
\newlength{\listespa}
\newcommand{\EnumLbl}[1]{\rm (#1)}%
\newcommand{\jsListe}[1]%
    {\noindent\makebox[\retraita][r]{\EnumLbl{#1}}%
     \hspace*{\listespa}\ignorespaces}
\newcommand{\tha}{\jsListe{a}}
\newcommand{\thb}{\jsListe{b}}

\newcommand{\pointn}{\noindent \makebox[1.2em]{$\bullet$}\ignorespaces}
\newcommand{\THFracs}[2]{\frac{#1}{#2}}%
 
\newcommand{\pqs}{\THFracs{p}{q}}

\newcommand{\pqps}{\THFracs{p'}{q'}}

\newcommand{\Indpq}[1]{#1_{\pqs}}%
\newcommand{\Indpqp}[1]{#1_{\pqps}}%
\newcommand{\Indtd}[1]{#1_{\tds}}%

\newcommand{\tds}{\THFracs{3}{2}}

\newcommand{\cts}{\THFracs{5}{3}}

\newcommand{\Tpq}{\Indpq{T}}
\newcommand{\Ltd}{\Indtd{L}}
\newcommand{\Lpqp}{\Indpqp{L}}

\newcommandy{\rhythm}[1][1={r}]{{\mathbf{#1}}}
\newcommandy{\rrhythm}[1][1={r}]{\overset{\raisebox{-0.1em}{\tiny$\circ$}}{\rhythm[#1]}}
\newcommandy{\labelling}[1][1={\lambda}]{{\boldsymbol{#1}}}
\newcommandy{\signature}[1][1={s}]{{\boldsymbol{#1}}}
\newcommandy{\slabelling}{\labelling[\sigma]}

\renewcommand{\mod}{\xmd{\text{\scriptsize\textbf \%}}\xmd}
\newcommand{\imod}{\xmd{\text{\tiny\textbf \%}}\xmd}

\newcommandy[o]{\val}[2][1=b]{%
  \ifthenelse{\equal{#2}{}}%
  {\pi}%
  {%
    \ifthenelse{\equal{#1}{b}}
      {\pi\hspace*{-0.1em}\left(#2\right)}%
      {\pi(#2)}%
  }
}
\newcommand{\fval}[1]{\operatorname{{\pi}_{#1}}}
\newcommand{\fvalp}{\fval{p}}

\newcommand{\evalp}[1]{\fvalp\hspace*{-0.1em}\left(#1\right)}

\newcommand{\fvalpqp}{\fval{\pqps}}
\newcommand{\evalpqp}[1]{\fvalpqp\hspace*{-0.1em}\left(#1\right)}
\newcommand{\Rep}[2]{\langle #1 \rangle_{#2}}
\newcommand{\pqRep}[1]{\Rep{#1}{\pqs}}
\newcommand{\pRep}[1]{\Rep{#1}{p}}

\newcommand{\rRep}[1]{\Rep{#1}{\rhth}}

\newcommand{\intint}[2]{\llbracket#1,#2\rrbracket}

\newcommandy{\pset}[2][1=n,2=R]{E_{#1}^{#2}}

\newcommandy{\rcod}[2][2=r]{\cod{#1}_{\rhythm[#2]}}

\newcommandWithStar{\Tr}{\mathcal{T}_{\rhythm}}
\newcommandWithStar{\taupq}{\Indpq{\tau}}
\newcommandWithStar{\pq}{\frac{p}{q}} 
\newcommandWithStar{\Lpq}{\Indpq{L}}
\newcommandWithStar{\Lpqbar}{\widebar{\Indpq{L}}}
\newcommandWithStar{\Lr}{L_{\rhythm}}
\newcommandWithStar{\Mr}{M_{\rhythm}}
\newcommandWithStar{\Kr}{K_{\rhythm}}
\newcommandWithStar{\Ar}{\mathcal{A}_{\rhythm}}
\newcommandWithStar[o]{\Ap}{A_p}
\newcommandWithStar{\Aps}{A_p^{\xmd *}}
\newcommandWithStar{\Aq}{A_q}
\newcommandWithStar{\Aqs}{A_q^{\xmd *}}
\newcommandWithStar{\negr}{\rhythm[s]_{\frac{p}{q}}}
\newcommandWithStar{\negL}{\widebar{\Lpq}}
\newcommandWithStar{\rpq}{\rhythm_{\frac{p}{q}}}
\newcommandWithStar{\canonlab}{\labelling{\gamma}_\frac{p}{q}}
\newcommandWithStar{\negl}{\labelling[\mu]_{\frac{p}{q}}}
\newcommandWithStar[o]{\N}{\mathbb{N}}
\newcommandWithStar[o]{\Nmu}{\mathbb{N}\setminus\set{0}}
\newcommandWithStar[o]{\Z}{\mathbb{Z}}
\newcommandWithStar[o]{\Tpq}{\Indpq{\mathcal{T}}}
\newcommandWithStar{\Tpqbar}{\widebar{\Tpq}}
\newcommandWithStar{\treer}{\rtree{\rhythm}}
\newcommandWithStar{\Lrl}{L_{\rhythm,\labelling}}

\newcommand{\rtree}[1]{\mathcal{J}_{#1}}
\newcommand{\qmu}{q-1}
\newcommand{\pmu}{p-1}
\newcommand{\iqmu}{{q-1}}

\newcommand{\ipmu}{{p-\text{\psscalebox{0.6}{$1$}}}}

\newcommand{\jpu}{{j~\!\!\text{\psscalebox{0.95}{$+$}\psscalebox{0.90}{$1$}}}}
\newcommand{\ijmu}{{j\text{\psscalebox{0.55}{$-$}\psscalebox{0.6}{$1$}}}}
\newcommand{\iimu}{{i\text{\psscalebox{0.55}{$-$}\psscalebox{0.6}{$1$}}}}
\newcommand{\ikmu}{{k\text{\psscalebox{0.55}{$-$}\psscalebox{0.6}{$1$}}}}

\newcommandx{\yaHelper}[2][1=\empty]{%
\ifthenelse{\equal{#1}{\empty}}%
  { \ensuremath{\scriptstyle{#2}}} 
  { \raisebox{ #1 }[0pt][0pt]{\ensuremath{\scriptstyle{#2}}}}  
}

\newcommandx{\yrightarrow}[4][1=\empty, 2=\empty, 4=\empty, usedefault=@]{%
  \ifthenelse{\equal{#2}{\empty}}
  { \xrightarrow{ \protect{ \yaHelper[#4]{#3} } } } 
  { \xrightarrow[ \protect{ \yaHelper[#2]{#1} } ]{ \protect{ \yaHelper[#4]{#3} } } } 
}

\newcommandy[o]{\pathx}[2][2={}]{
  \nlb%
  \ifthenelse{\equal{#2}{}}%
    {\yrightarrow{\ #1\ }[-1pt]}%
    {\underset{#2}{\xrightarrow{\ #1 \ }}}%
  \nlb%
}

\newcommand{\vmcustomthmname}{Theorem}

\newcommand{\itree}{i-tree\xspace}
\newcommand{\itrees}{i-trees\xspace}
\newcommand{\vmoverbracket}[1]{\overbracket[0.6pt][2pt]{#1}}

\newcommand{\OBK}[1]{\vmoverbracket{#1}}

%% file: RBR-int.tex

\section{Introduction}

In this work, we study a family of (infinite) trees and languages 
that are defined by means of a new technique: the \emph{breadth-first 
search} that we have introduced in a recent paper~\cite{MarsSaka2014b}. 
In that paper, we have explained that an ordered tree of finite
degree~$\Tc$ can be characterised by the infinite sequence of the
degrees of its nodes visited in the order given by the breadth-first 
search, called the \emph{signature}~$\sigs$ of~$\Tc$.
This signature~$\sigs$, together with an infinite sequence~$\lablbd$ 
of letters taken in an ordered alphabet characterises then a 
\emph{labelled}  
tree~$\Tc$.

If the sequence~$\lablbd$ is consistent with the signature~$\sigs$ 
--- we call the pair~$(\sigs,\lablbd)$ a \emph{labelled signature} 
--- the breadth-first search  
of~$\Tc$ corresponds to the enumeration in the \emph{radix order} of 
the prefix-closed language~$L_{\Tc}$ of branches of~$\Tc$.
And we have shown (\cite[Th.~1]{MarsSaka2014b}) that regular trees or
(prefix-closed) regular languages are characterised by those labelled 
signatures that are substitutive sequences.

Here, we consider and study  the simplest possible 
signatures, and labelled signatures, when seen as infinite words, 
namely the purely periodic ones.
The labelled tree --- call it~$\Trtd$ --- shown at \figur{l32} gives 
an example of a labelled  
tree having such a  periodic labelled signature.
The nodes of~$\Trtd$ are numbered by integers in the order of a 
breadth-first search and, with exception of the root~$0$ for a reason 
that will be explained later, even nodes have two children and odd 
nodes one, which results in the sequence
$\msp 2,1,2,1, \ldots = 2\xmd 1^{\omega}\msp$
for the signature.
Moreover,  the sequence of labels of the arcs in the same 
breadth-first search is 
$\msp 0,2,1,0,2,1, \ldots = 0\xmd 2\xmd 1^{\omega}\msp$.

The branch language of~$\Trtd$, that is, the set of words that label 
the paths from the root to every node, is the language~$\Ltd$ of the 
representations of the integers in the so-called \emph{numeration 
system in base~$\tds$} that has been introduced and studied 
in~\cite{AkiyEtAl08}.

\begin{figure}[ht]
	\medskipneg 
  \centering
  \includegraphics[width=\scaled\linewidth]{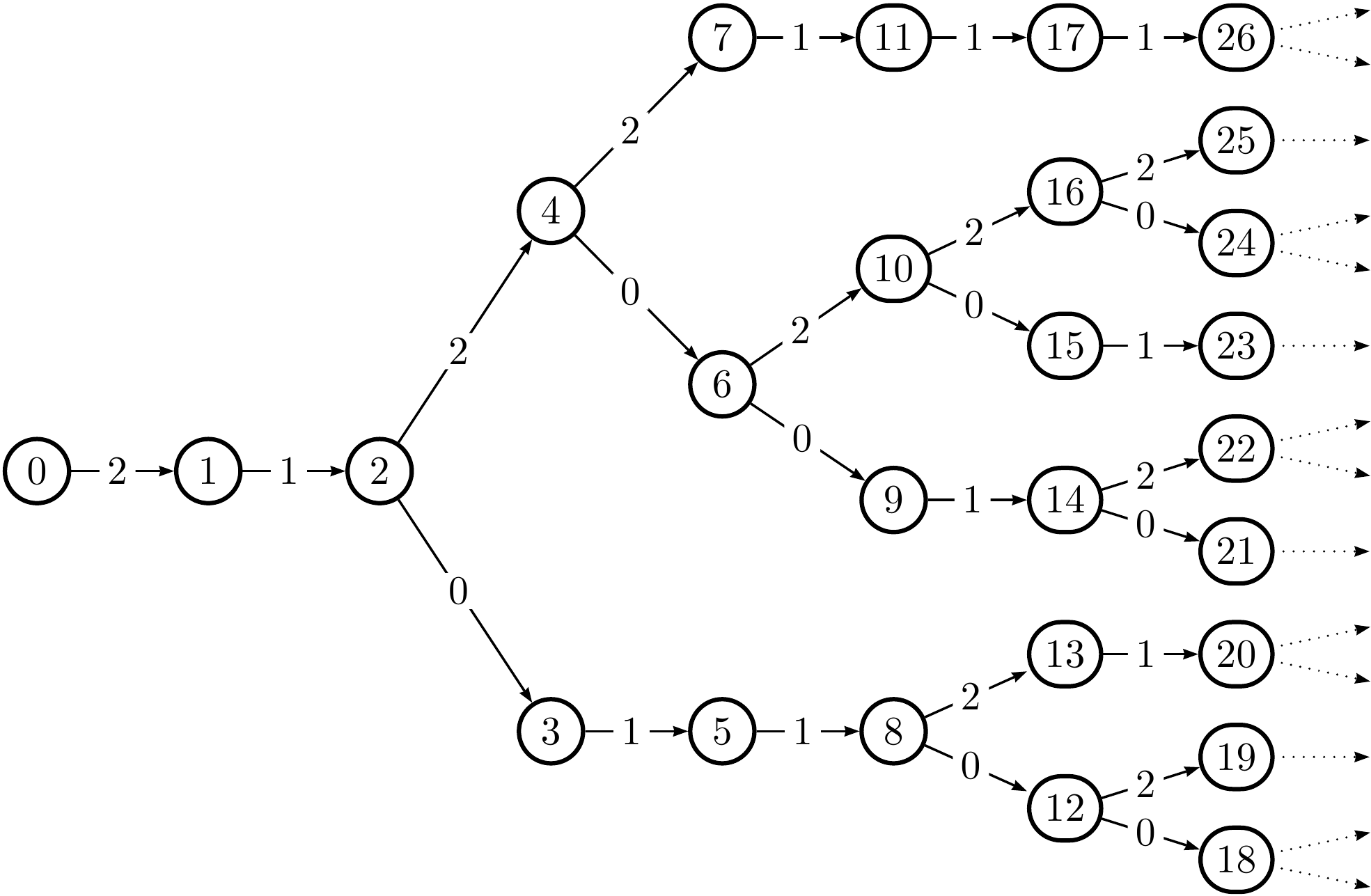}
  \caption{The tree~$\Trtd$, representation of the language~$\Ltd$}
  \lfigure{l32}
	\medskipneg 
\end{figure}

This language~$\Ltd$ and more generally the languages~$\Lpq$ of the 
representations of the integers in base~$\pqs$ are indeed the 
starting point of this work.
It is a challenge to better understand their structure, both from a 
number theoretic point of view --- as we did in~\cite{AkiyEtAl08} --- 
and from a formal language theory point of view as they seem to be at 
the same time `very regular' and completely orthogonal to all 
classification and usual tools of classical formal language theory.
In particular, none of them are regular languages and they even defeat 
any kind of iteration lemma. 
In a former work we have shown that these languages enjoy a kind of 
`autosimilarity property'~\cite{AkiyEtAl13}.
Here, we show that they are somehow characterised by their periodic 
labelled signature.

First, we show that the labelled signature of these languages are 
periodic (\theor{lpq-rhythm-perm}).
Stating the converse requires some more definitions.

We call \emph{rhythm} of \emph{directing parameter}~$(q,p)$ a 
$q$-tuple of integers whose sum is~$p$: $\rhth =\rhthtp$.
We then describe how such rhythm~$\rhth$ allows to
generate a tree such that the $n$-th 
node (in the breadth-first order) has~$r_{k}$ successors, 
where~$k$ is congruent to~$n$ modulo~$q$ (plus a special rule for the root).

The main result of this paper reads then (below we give the definition 
of FLIP languages, which, roughly speaking, are languages that meet 
no kind of iteration lemma):

\begin{customthm}{\protect\rtheorem*{trans-language}}
	Let \Kr* be the branch language of the tree 
 generated by a rhythm~$\rhythm$ of directing parameter~$(q,p)$.
  \begin{enumerate}[label=\alph{*}.]
    \item 
      If \pq* is an integer, then \Kr* is a regular language;
    \item 
      If \pq* is not an integer, then \Kr* is a FLIP language.
\end{enumerate}
\end{customthm}

In the general case, the rhythm of \Lpq* corresponds to \emph{the most equitable 
  way of parting~$p$ objects into~$q$ cases} 
  (with a bias to the left when necessary).
We call it the \emph{Christoffel rhythm} associated with~\pq*,
as
it can be derived from the more classical notion of
  Christoffel word of slope~$\frac{p}{q}$ (\cf~\cite{BersEtAl08}), that is,
  the canonical way to approximate the line of slope~$\frac{p}{q}$
  on a~$\Z\times\Z$ lattice.

\medskip

The proof of \rtheorem{trans-language} is the purpose of \rsection{red-to-rat}
  and consists of the reduction of any structure generated by a rhythm
  to the number system whose base is the \gr of this rhythm.
In fact, the language generated by a rhythm is simply a non-canonical 
  representation of the integers in this base, in the sense that the integers 
  are represented on a non-canonical alphabet.
Using the existing work on alphabet conversion in rational base number 
  systems (\cf \cite{AkiyEtAl08} or \cite{FrouSaka10hb}) it allows to
  conclude that both languages are basically as complicated (or as simple, 
  in the degenerate case where the \gr happens to be an integer).
  
\medskip

This article is organised as follows.  In the preliminaries, we
present the three notions used in the sequel: the numeration system in
base~$\frac{p}{q}$, the Finite Left Iteration Property, and the trees
and their signature.
In \rsection{rhy-tre-lan} we give a precise definition of the
breadth-first generation of infinite trees and language by a rhythm.
Then in \rsection{rhy-gen}, we describe how this process can be used
to generate the language of the representation of integers in a
rational base numeration system.
Finally, in \rsection{red-to-rat}, we prove that any language build by
a rhythm is in some sense a non-canonical representation of the
integers in some underlying rational base.

%


%% file: RBR-sec2.tex
\section{Preliminaries and Notation}
\lsection{prelim}%

Given two \emph{positive integers}~$n$ and~$m$, 
  we denote by~$\frac{n}{m}$ their division in~$\Q$; by~$n\div m$ 
  and~$n \mod m$ respectively the quotient and 
  the remainder of the Euclidean division of~$n$ by~$m$, 
  that is,~$\msp {n=(n\div m)\xmd m + (n\mod{m})} \msp$ 
  and~$\msp 0\leq (n\mod{m}) < m$.
Additionally, we denote by~$\intint{n}{m}$ the integer 
  interval~${\set{n,(n+1),\ldots, m}}$.

\subsection{Rational Base Numeration Systems}

Let~$p$ be an integer, $p\geq2$, and~$\Ap=\intint{0}{\pmu}$ the 
alphabet of the~$p$ first digits.
Every word~$w=a_n\xmd a_{n-1} \cdots \xmd a_0$ of~$\Aps$ is given a 
value in~$\N$ by the \emph{evaluation function}~$\fvalp$:
%
\begin{equation}
\evalp{a_n\xmd a_{n-1}\cdots\xmd a_0}\e =\e\sum^{n}_{i=0} a_i\xmd p^i
\eqvrg
\notag
\end{equation}
and~$w$ is a $p$-development of~$n$.
Every~$n$ in~$\N$ has a unique \emph{$p$-development} without 
leading~$0$'s in~$\Aps$: it is called the \emph{$p$-representation} 
of~$n$ and is denoted by~$\pRep{n}$.
The \emph{$p$-representation} of~$n$ can be computed from 
left-to-right by a greedy algorithm, and also from 
\emph{right-to-left} by iterating the Euclidean division of~$n$ 
by~$p$, the digits~$a_{i}$ being the successive remainders.
The language of the $p$-representations of the integers is the 
rational language
$\msp  L_{p} = \Defi{\pRep{n}}{n\in\N} 
             = \left(\Ap\bk 0\right)\xmd\Aps\msp$.

\medskip 

Let~$p$ and~$q$ be two co-prime integers,~$p>q>1$.
In~\cite{AkiyEtAl08}, we have generalised these classical, and 
obvious, statements to the the more exotic case of \emph{numeration 
system with rational base~$\pq$}.
Given a positive integer~$n$, let us define~$N_0=N$ and, for 
all~$i>0$,
\begin{equation}
  q\xmd N_i = p\xmd N_{(i+1)} + a_i 
  \eqvrg
\lequation{mod-euc-div}
\end{equation}
  where~$a_i$ is the remainder of the Euclidean 
  division of $q\xmd N_i$ by~$p$,
  hence in~${\Ap=\intint{0}{p-1}}$.
Since~$p>q$, the sequence~$(N_i)_{i\in\N}$ is strictly decreasing and 
  eventually stops at~$N_{k+1}=0$. 
Moreover, it holds that
\begin{equation}
  N = \sum^{k}_{i=0} \frac{a_i}{q} \left( \pq \right)^i 
  \eqpnt
  \notag
\end{equation}

The \emph{evaluation function}~$\pi_{\pq}$ is derived from this formula. 
Given a word~${a_na_{n-1}\cdots a_0}$ over \Ap*, and indeed over any 
  alphabet of digits,  its \emph{value} is defined by
\begin{equation}\label{eq.pi}
  \pi_{\pq}(a_na_{n-1}\cdots a_0) = 
  \sum^{n}_{i=0} \frac{a_i}{q} \left( \pq \right)^i
  \eqpnt 
\end{equation}

Conversely, a word~$u$ in~$\Aps$ is called a \pq*-\emph{representation} of
  an integer~$x$ 
  if~${\val{u}=x}$.
Since the representation is unique up to leading 0's
  (see~\cite[Theorem~1]{AkiyEtAl08}) the 
  \pq*-representation of~$x$ which does
  not starts with a 0 is is denoted by~$\cod{x}_{\pq}$
  and 
  can be computed with the 
  modified Euclidean division algorithm above.
By convention, the representation of 0 is the empty word~$\epsilon$.
The set of $\pq$-representations of integers is denoted by~$\Lpq$:
\begin{equation}
  \Lpq = \Defi{\pqRep{n}}{n\in\N} \eqpnt
  \notag
\end{equation}

It is immediate that~$\Lpq$ is prefix-closed (since, in the modified Euclidean
  division algorithm~$\cod{N}=\cod{N_1}.a_0$) and right-extendable
  (for every representation~$\cod{n}$, there exists (at least) an~$a$ 
  in \Ap* such 
  that~$q$ divides~$(np+a)$ and 
  then~${\cod{\frac{np+a}{q}}=\cod{n}.a}$).
As a consequence,~$\Lpq$ can be represented as an infinite tree;
  it is shown at \rfigure{l32}, in the introduction.
By abuse of language, in the following we will write that~$n \pathaut{u}{\Lpq} m$ 
  (or~$n \pathx{u} m$, for short) if~$\cod{m}=\cod{n}.u$;
  it should be noted that, with this notation, the following equation hold.
\begin{equation}
\lequation{pathlpq}%
\forall n \in \N \quantvrg 
\forall m \in \N \quantvrg 
\forall a \in \Ap \quantsp
  n \pathx{a} m \e \iff \e a = q\xmd m - p \xmd n 
\eqpnt 
\end{equation}

It is known that \Lpq* is not rational and not even 
  context-free~(\cf \cite{AkiyEtAl08}).
In fact \Lpq* defeats any reasonable kind of pumping lemma; 
  it possesses the \emph{Finite Left Iteration Property}, 
  discussed in the next \rsection{blip}.

\begin{remark}

Even though we separated the notions of rational and integer base number
  systems in order to give specific statements, it should be noted
  that the former extends naturally the latter.
Indeed, in the case where~$q=1$, the definitions 
  of~$\val{}_{\pq}$,~$\cod{\,\cdot\,}_{\pq}$ and \Lpq* respectively coincide 
  with those of~$\val{}_p$,~$\cod{\,\cdot\,}_p$ and~$L_p$.
In the sequel, we will consider the base \pq* such that~$p>q\geq 1$, 
  that is indifferently one number system or the other.
\end{remark}

\begin{remark}
It should be noted that a rational base number systems is \emph{not} 
  a~$\beta$-numeration --- where the representation of a number is 
  computed by the (greedy) R\'enyi algorithm 
  (cf.~\cite[Chapter~7]{Loth02}) --- in the special case
  where~$\beta$ is a rational number.
In such a system, the digit set is~$\set{0,1, \ldots, \ceil{\pq}}$
  and the weight of the~$i$-th leftmost
  digit is~$(\pq)^i$; whereas in the rational base number system, 
  they are~$\set{0,1~\ldots(p-1)}$ and~$\frac{1}{q}(\pq)^i$ respectively.
\end{remark}

\subsection{The Finite Left Iteration Property (FLIP)}
\lsection{blip}%
               
We define here a strong `non-iteration' property of languages that 
will be closely related to the breadth-first generation process (by 
rhythm) later on.\footnote{%
   We have introduced this property in \cite{MarsSaka13b} under the
   name of \emph{Bounded Left Iteration Property}, or \emph{BLIP} for
   short.  The term `bounded' is indeed improper according to usual
   terminology and was a mistake.}

\begin{definition}\ldefinition{blip}
  A language~$L$ of~$\Ae$ has the \emph{Finite 
    Left Iteration Property}, or is a \emph{FLIP language} for short, 
	if for all~$u$,~$v$ in~$\N$, $|v| \geq 1$,
  \begin{equation*}
	\text{$u\xmd v^i$ is prefix of a word of~$L$ for only
	\strong{finitely} many~$i$ in~$\N$.}
  \end{equation*}
\end{definition}

Clearly, a FLIP language is neither regular, nor context-free.
In~\cite{AkiyEtAl08}, it has been shown that the languages~$\Lpq$, for 
all coprime~$p$ and~$q$, are FLIP languages. 

Although, or because, the Finite Left Iteration Property is the 
strongest way of contradicting any kind of iteration lemma, it 
is difficult to find natural examples of FLIP languages in the 
classical formal language theory.

The set of the prefixes of any infinite aperiodic word is a FLIP 
language, but this is rather trivial an example (as a language) since 
the number of words of every length is~$1$.

\begin{example}
        \lexample{blip-trivial}
Let~$L_{\text{fibo}}=\Defi{\phi^i(0)}{i\in\N}$ be the language of 
Fibonacci words, defined by the Fibonacci morphism~$\phi$:
$\phi(0) = 01$ and~$\phi(1) = 0$.
As the (infinite) Fibonacci word is power 4-free,~$L_{\text{fibo}}$ 
has no prefix of the form~$u\xmd v^4$ and is a FLIP language.
Since the power, that is~4, is independent of~$u$ 
  and~$v$,~$L_{\text{fibo}}$ a fortiori possess the \emph{Bounded} 
  Left Iteration Property.
\end{example}

The {Finite Left Iteration Property} is related  by the following 
statement to another property of language theory called `IRS'  
  (for \emph{Infinite Regular Subset}, \cf~\cite{Greibach1975}): 
a language is IRS if it does not contains any rational sublanguage.

\begin{proposition}
\lproposition{flip-equiv}%
For a language~$L$, the following statements\footnote{%
    $\pref{L}$ denotes the closure of~$L$ by prefix:
	$\pref{L}=\Defi{u\in\Ae}{\ext v\in\Ae\quantsmsp u\xmd v\in L}$.} 
are equivalent:
\begin{enumerate}[label=(\roman{*})]
  \item $L$ is a FLIP language.
  \item $\pref{L}$ is IRS.
  \item The topological closure of~$L$ contains aperiodic (infinite) words only.
\end{enumerate}
\end{proposition}
\begin{proof}
(i) $\Rightarrow$ (ii)\e 
If~$\pref{L}$ contains an infinite rational sublanguage, it contains a
subset~$u\xmd v^*w$ for some words~$u$,~$v$ and~$w$, hence, for
infinitely many integers~$i$,~$u\xmd v^i$ is a prefix of some words
of~$L$, a contradiction.
  
{(ii) $\Rightarrow$ (iii)}\e
Let us assume that~$w=u\xmd v^{\omega}$ belongs to the topological 
closure of~$L$. 
It implies that, for every integer~$i$,~$u\xmd v^i$ is the prefix of a word of~$L$,
hence~$u\xmd v^*$ is a sublanguage of~$\pref{L}$, a contradiction. 

The proof of the implication (iii) $\Rightarrow$ (i) is analogous.
\end{proof}

\begin{lemma}
\llemma{flip-stable}
  The class of FLIP languages is stable 
    by finite union, arbitrary intersection, sublanguage,
    concatenation and inverse morphism image.
\end{lemma}

\begin{proof}[Proof for concatenation]
Let~$L$ and~$M$ be two FLIP languages.
Let~$u$ and~$v$ be two words, and~$J$ the set of integers~$j$ such 
that the prefix~$w_j$ of length~$j$ of~$u\xmd v^\omega$ is in~$L$.
Since~$L$ is FLIP language,~$J$ is finite.
For every~$j$ in~$J$, the ultimately periodic 
word~$(w_j)^{-1}u\xmd v^{\omega}$ is not in the topological closure 
of~$M$, hence there are only finitely many~$k$ such 
that~$(w_j)^{-1}u\xmd v^{k}$ is a prefix of a word of~$M$.
Summing up for all~$j$ in~$J$, there are only finitely many~$k$ such 
that~$u\xmd v^{k}$ is a prefix of a word of~$L\xmd M$.
\end{proof}

It is easy to give examples showing that the class of FLIP languages 
is not closed under complementation, star, transposition 
  and direct morphism image.

\subsection{On Trees and Signatures}
\lsection{ontrees}

Classically, trees are undirected graphs in which any two vertices are 
  connected by exactly one path (\cf \cite{Dies97}, for instance).
Our view differs in two respects.

First, a tree is a \emph{directed} graph~$\Tc=(V,\Gamma)$ such that there exist a 
  \emph{unique} vertex, called \emph{root}, which has no incoming arc,
    and there is a \emph{unique (oriented) path} from the root to every 
    other vertex.
Elements of the tree~$\Tc$ gets particular names: 
  vertices are called \emph{nodes}; 
  if~$(x,y)$ is an arc,~$y$ is called \emph{a child} of~$x$ and~$x$ 
  \emph{the father} of~$y$; 
  a node without children is a \emph{leaf}.
We draw trees with the root on the left, and arcs rightwards.

Second, our trees are \emph{ordered}, that is, that there is a total order 
  on the set of children of every node. 
The order will be implicit in the figures, with the convention that lower 
  children are smaller (according to this order).

\medskip 

It will prove to be extremely convenient to have a slightly different
look at trees and to consider that the root of a tree is also a
\emph{child of itself} that is, bears a loop onto itself.
This convention is sometimes taken when implementing tree-like
structures (for instance the unix/linux file system).
We call such a structure an \emph{\itrees}.
It is so close to a tree that we pass from tree to \itree (or
conversely) with no further ado.
\figur{l32} shows a tree and \figur{a-it32} shows the associated 
i-tree. 

\medskip

The degree of a node is the number of its children.
In the sequel, we consider infinite ordered (i-)trees of finite 
degree, that is, all nodes of which have finite degree.
The breadth-first search of such a tree defines a total ordering of 
its nodes.
We then consider that the set of nodes of an (i-)tree is always the 
set of integers~$\N$.
The root is~$0$ and~$n$ is the $(n+1)$-th node visited by 
the search.

Let~$\Tc$ be an ordered (i-)tree of finite degree.
The sequence~$\sigs$ of the degrees of the nodes of~$\Tc$ visited in 
the breadth-first search of~$\Tc$ is called the \emph{signature} 
of~$\Tc$ and is \emph{characteristic} of~$\Tc$, that is, one can 
compute, or build, $\Tc$ from~$\sigs$.
By convention, and whether~$\Tc$ be a tree or an i-tree, \emph{the
signature is always that of the i-tree}.


%% file: RBR-sec3.tex
\section{Rhythmic trees and languages}
\lsection{rhy-tre-lan}%

\subsection{Rhythms and their geometric representation}
\lsection{rhythm}%

\begin{definition}
\ldefinition{rhythm-growth}%
Let~$p$ and~$q$ be two integers with~$p>q\jsgeq 1$.
  
\begin{enumerate}

\item  
We call \emph{rhythm} of directing parameter~$(q,p)$, 
a \emph{$q$-tuple}~$\rhth$ of non-negative integers \emph{whose sum 
is~$p$}: 

\medskipneg 
\begin{equation}
\rhth = \rhthtp 
\e\text{and}\e
\sum_{i=0}^{\qmu} r_i = p
\eqpnt
\notag
\end{equation}
   
\medskipneg 
\item  
We say that a rhythm~$\rhth$ is \emph{valid} if it satisfies the 
following equation:

\medskipneg 
\begin{equation}
  \forall j \in \intint{0}{\qmu} \quantsp 
  \sum_{i=0}^{j} r_i ~>~ \jpu
  \eqpnt
  \label{eq.validrhythm}
\end{equation}
      
\medskipneg 
\item  
We call \emph{\gr}{ }of~$\rhythm$ the rational number~$z=\frac{p}{q}$,
also written~$z=\frac{p'}{q'}$ where~$p'$ and~$q'$ are the quotients
of~$p$ and~$q$ by their greatest common divisor (gcd), hence coprime.

\end{enumerate}
\end{definition}

The \gr of a rhythm is always greater than~$1$.
Examples of rhythms of \gr~$\frac{5}{3}$ are
  $(2,2,1)$, 
  $(3,0,2)$, 
  $(1,2,2)$, 
  $(2,2,1,2,2,1)$, 
  $(2,1,3,0,0,4)$;
  all but the third one are valid;
  the directing parameter is (3,5) for the first three, 
  and (6,10) for the last two.
  
\medskip 

Rhythms are given a very useful geometric 
  representation as \emph{paths} in the $\Z\x\Z$-lattice and such paths are
  coded by \emph{words} of~$\alphxye$ where~$x$ denotes a unit horizontal
  segment and~$y$ a unit vertical segment.
Hence the name \emph{path} given to a \emph{word} associated with a rhythm.

\begin{definition}
\ldefinition{rhy-pat}%
With a rhythm~$\rhth=\rhthtp$ of directing parameter~$(q,p)$, 
we associate the word~$\pthw{\rhth}$ of~$\alphxye$:

\smallskipneg 
\begin{equation}
\pthw{\rhth}=  
   y^{r_0}x\xmd y^{r_1}x\xmd y^{r_2} \cdots x\xmd y^{r_{\iqmu}}x
   \notag
\end{equation}
which corresponds to a path from~$(0,0)$ to~$(q,p)$ 
    in the~$\Z\x\Z$-lattice.
\end{definition}

\figur{rhy-pat} shows the paths associated with three of the above
rhythms.  
It then appears clearly that \defin{rhythm-growth}.2
can be restated as 
`a rhythm is valid if the associated path is strictly above 
the line of slope 1'.

\begin{figure}[ht!]
\e
  \subfloat[rhythm (3,1,1)]%
     {\lfigure{christ-311}%
          \includegraphics[width=\scalet\linewidth]{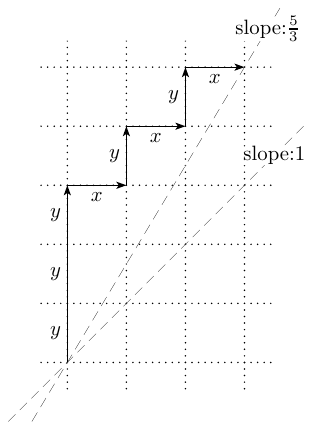}}%
  \hfill
  \subfloat[rhythm (2,2,1)]%
     {\lfigure{christ-221}%
         \includegraphics[width=\scalet\linewidth]{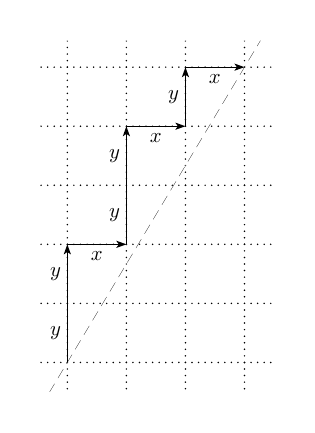}}
  \hfill%
  \subfloat[rhythm (1,2,2)]
     {\lfigure{christ-122}%
      \includegraphics[width=\scalet\linewidth]{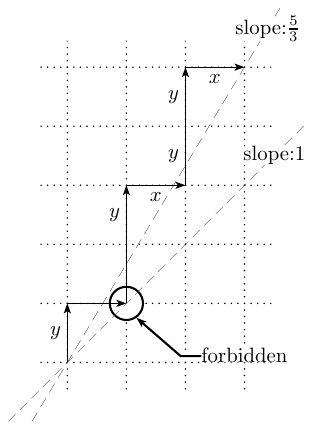}}%
\e
  \caption{Words and paths associated with rhythms of directing parameter (5,3)}
  \lfigure{rhy-pat}
\end{figure}

\subsection{Generating trees by rhythm}
\lsection{tree}%

As said above, we have describe in~\cite{MarsSaka2014b} the procedure 
that reconstructs a tree~$\Tc$ from its signature~$\sigs$.
We present here this construction in the case where~$\sigs$ is a purely 
periodic word~$\sigs=\rhth^{\omega}$. 
First, a static description of the result.

\begin{definition}
\ldefinition{trans-tree}%
\nopagebreak%
Let~$\rhth=\rhthtp$ be a (valid) rhythm.
The tree~$\Trr$ generated by~$\rhth$ is defined by:
\begin{itemize}
    \item the root~$0$ has~$(r_0 -1)$ children: the nodes~$1$, $2$,\ldots, 
	and~$(r_0 -1)$; 
    \item for every~$n>0$, the node~$n$ has~$r_{n\imod{q}}$ 
	children: the nodes~$(m+1)$, $(m+2)$, \ldots, and~$(m+r_{n\imod{q}})$, 
	where~$m$ is  the greatest child of the node~$(n-1)$.
\end{itemize}
\end{definition}  

\figur{t311} shows~$\Trtuu$, \figur{l32} shows~$\Trdu$ (if one forgets the 
labels on the arcs).

\begin{figure}[ht!]%
\e
  \subfloat[The tree~$\Trtuu$]%
     {\lfigure{t311}%
      \includegraphics[width=\scaleq\linewidth]{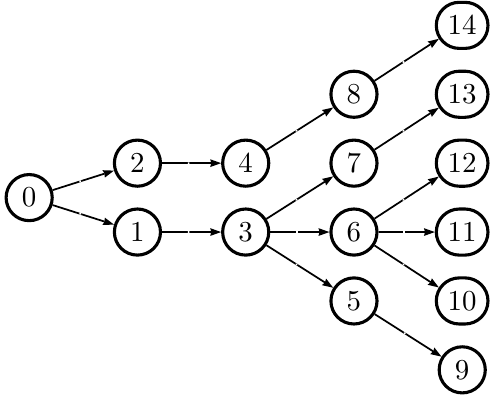}}
  \hfill
  \subfloat[The language~$\Klgtuu$]
     {\lfigure{k311}%
      \includegraphics[width=\scaleq\linewidth]{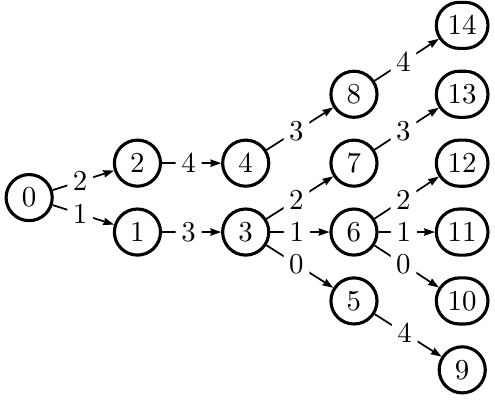}}
\e
  \caption{Tree and language generated by the rhythm~$(3,1,1)$}
  \lfigure{k311-tot}
\end{figure}

\begin{figure}[ht]
  \centering
  \includegraphics[width=\linewidth]{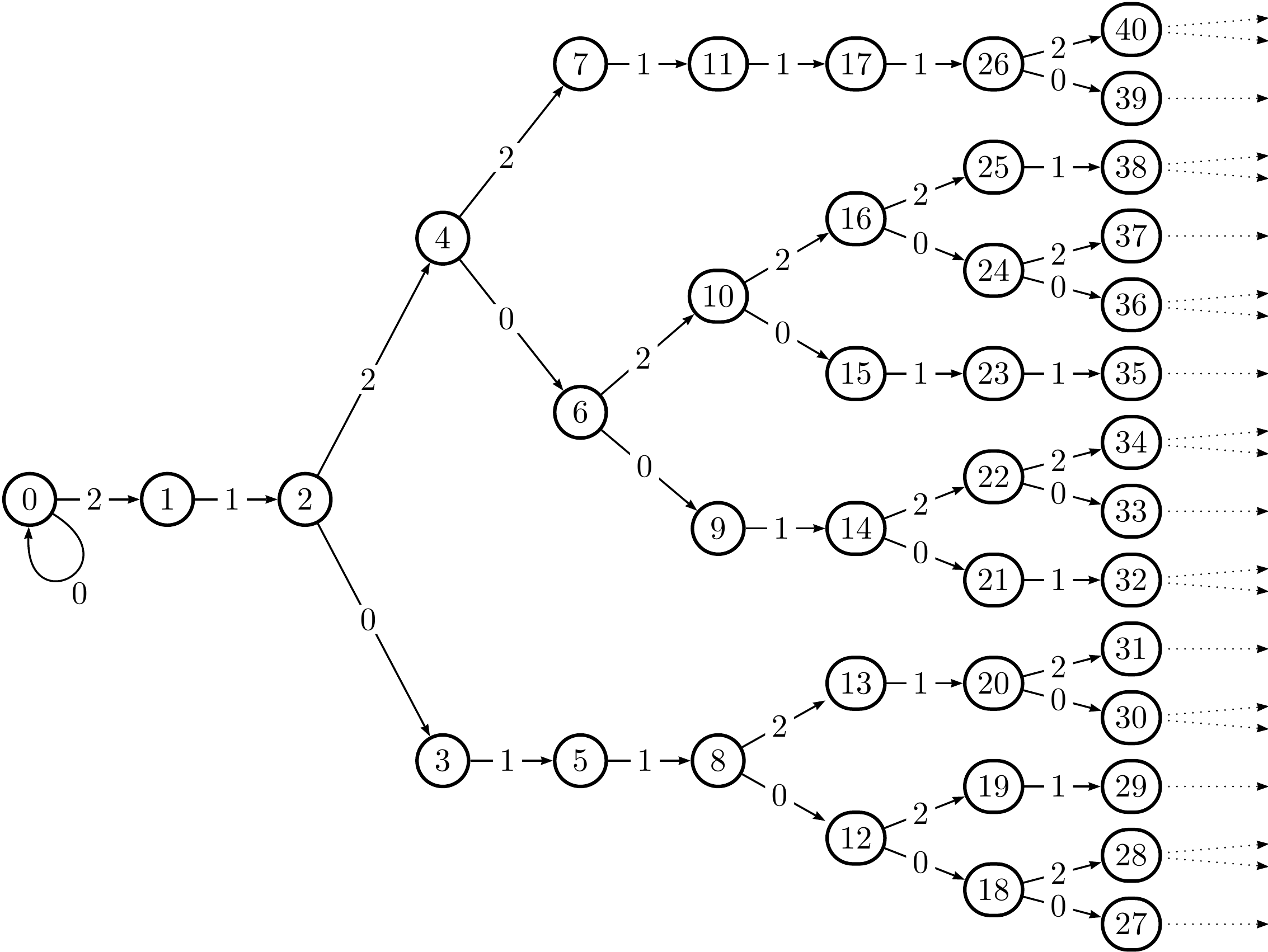}
  \caption{The i-tree associated with~$\Trtd$}
  \lfigure{a-it32}
\end{figure}

\begin{figure}[ht!]
  \centering  
  \includegraphics[width=\linewidth]{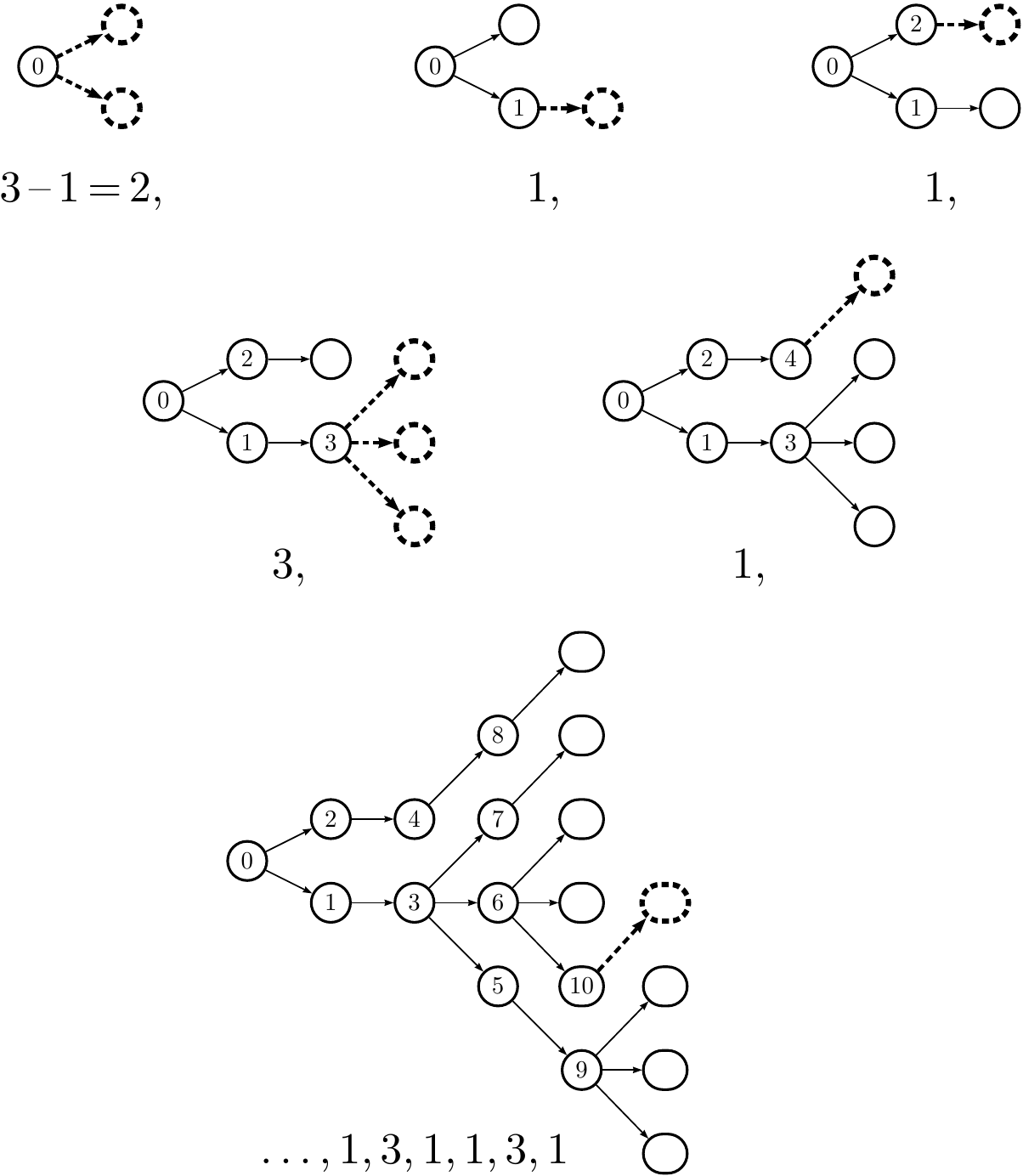}
  \caption{Building~$\Trtuu$ from the rhythm~$(3,1,1)$} 
  \lfigure{a-steps-Trtuu}
\end{figure}
\smallskip 

If we had treated the root~$0$ in the same way as the other nodes, we 
would have said that~$0$ has~$r_{0}$ children: the nodes~$0$, $1$, 
$2$,\ldots, and~$(r_0 -1)$ and we would have  defined the 
\emph{i-tree} (associated with)~$\Trr$. 
For instance, the i-tree~$\Trdu$ is shown at \figur{a-it32}.

The description of a procedure that builds~$\Trr$ from~$\rhth$ gives 
a more dynamical view on the process.
The procedure maintains two integers,~$n$ and~$m$, both initialised 
to~$0$: 
$n$~is \emph{the node to be processed} and~$m$ is \emph{the next node 
to be created}.
At every step of the procedure,~$r_{(n\imod q)}$ nodes are created:
the nodes~$m$, $(m+1)$, \ldots, and~$(m+r_{(n\imod q)}-1)$ 
and~$r_{(n\imod q)}$ arcs are created, from~$n$ to every new node.
Then~$n$ is incremented by~1, and~$m$ by~$r_{(n\imod q)}$.

This procedure indeed builds the \emph{i-tree}~$\Trr$ since its first 
step creates a loop~$0 \pathx{} 0$.
The \emph{tree}~$\Trr$ is obtained by removing this loop, and the
root~$0$ has then~$(r_0 -1)$ children only.
%
\figur{a-steps-Trtuu} shows the first five and the tenth steps of the 
procedure for the rhythm~$(3,1,1)$.
It is an easy verification that the tree built by that procedure 
meets the tree defined at \defin{trans-tree}.

The \emph{validity} of the rhythm is the necessary and sufficient 
condition for~$m$ always be greater than~$n$ in the course of the 
execution of the procedure, that is, a node is always `created' 
before being `processed', or, equivalently, for the tree described at 
\defin{trans-tree} be infinite.

A direct consequence of the building of~$\Trr$ by the procedure is 
that~$q$ consecutive nodes of~$\Trr$ (in the breadth-first search) 
have~$p$ (consecutive) children, hence the name \emph{growth} given to 
the ratio~$\pq$.
More precisely, the following holds.

\begin{lemma}
\llemma{growth}%
Let~~$\Trr$ be the tree generated by the rhythm~$\rhth$ of directing 
parameter~$(q,p)$.
Then, for all~$n$, $m$ in~$\N$:
\begin{equation}
n\pathaut{}{\Trr}m
\e\iff\e
(n+q)\pathaut{}{\Trr}(m+p)
\eqpnt 
\notag
\end{equation}
\end{lemma}

\subsection{Labelling of Rhythmic Trees}
\lsection{lab-rhy-tre}%

An ordered tree~$\Tc$ defines its signature~$\sigs$, and we have seen how to 
reconstruct~$\Tc$ from~$\sigs$ (at least in the case 
where~$\sigs=\rhth^{\omega}$).
In the same way, a labelled ordered tree~$\Tc$ defines the 
sequence~$\lablbd$ of the labels of the arcs as they are visited in the 
breadth-first search, and~$\Tc$, and hence its branch 
language~$L_{\Tc}$, will be determined by the pair~$(\sigs,\lablbd)$.

The (finite) alphabet of labels is ordered as well --- we consider 
the case of digit alphabets only.
Of course, we want the labelling of~$\Tc$ be consistent with the 
order of~$\Tc$, that is, the breadth-first search of~$\Tc$ yield the 
\emph{radix order} on~$L_{\Tc}$, which is equivalent to the condition 
that the children of every node~$n$ are in the same order as the 
labels of the arcs that come from their father~$n$.

We consider here periodic signatures~$\sigs=\rhth^{\omega}$ 
where~$\rhth$ is a rhythm of directing parameter~$(q,p)$.
We then will consider pairs~$(\sigs,\lablbd)$ 
with~$\lablbd=\labgmm^{\omega}$ where~$\labgmm$ is a sequence of 
letters (digits) \emph{of length~$p$}.
And we say that $\Tc$, and its branch language~$L_{\Tc}$, 
are determined by the pair~$(\rhth,\labgmm)$. 

It follows from \lemme{growth} that the labelling is consistent on the 
whole tree if and only if it is consistent on the first~$q$ nodes, 
hence on the first~$p$ arcs, in which case we say that~$\labgmm$ is 
\emph{valid}.
A first, and obvious, valid labelling is the sequence~$\labgp$ of the 
first~$p$ digits: 
$\msp\labgp=\labgptp\msp$, which we call the \emph{naive labelling}.

\begin{definition}
\ldefinition{kr-nai-lab}%
Let~$\rhth$ be a rhythm of directing parameter~$(q,p)$ 
and~$\labgp$ the naive labelling.
We denote by~$\Klgr$ the branch language of the labelled tree 
determined by the pair~$(\rhth,\labgp)$, that is, the tree~$\Trr$ 
labelled by:
\begin{equation}
    \forall n,m \in \N \quantsp
	n \pathaut{a}{\Trr} m \e \text{with } a=m\mod{p} 
	\eqpnt
    \notag
\end{equation}
\end{definition}

\figur{k311} shows the language~$\Klgtuu$ while~$\Klg{(3,1)}$ is 
shown at \figur{k31a}.

\smallskip 

All elements of the main result of this paper 
are now defined: it states that the language built by a rhythm and the
naive labelling is either regular or FLIP, according to whether the
\gr of the rhythm is an integer or not.

\begin{theorem}
\ltheorem{trans-language}%
Let~$\rhth$ be a rhythm of directing parameter~$(q,p)$.

\nopagebreak

\begin{enumerate}[label=\alph{*}., %
                  ref=\protect\rtheorem*{trans-language}\alph{*}]
    \item \ltheorem{trans-lang-rat}
      If \pq* is an integer, then \Kr* is a regular language.
    \item \ltheorem{trans-lang-blip}
      If \pq* is not an integer, then \Kr* is a FLIP language.
\end{enumerate} 
\end{theorem}

The proof of the whole statement consists in a reduction to the case
of the representation language of rational base numeration systems and
occupies indeed the remainder of the paper (\theor{lpq-rhythm-perm}
and \theor{lr-num-pos}).
However, \theor{trans-language}.a can be established by a direct and
simpler proof given below; it is
also a corollary of the main result
of~\cite{MarsSaka2014b}.
\begin{proof}[Proof of \rtheorem{trans-lang-rat}]
The proof consists in considering the underlying tree~$\Trr$
of~$\Klgr$ as an \emph{infinite} automaton and then proving that it
has a finite number of classes in the Nerode equivalence.  
More precisely, we prove that two states~$n$ and~$m$, with~$n$ and~$m$
strictly positive, are Nerode-equivalent if they are congruent
modulo~$q$.

For every integer~$i$, we denote by~$\sim_i$ the following 
equivalence relation:
  given two states~$n$ and~$m$, we write that~$n\sim_i m$ 
  if, for all word~$u$ of \Aps* of length~$i$, $n\cdot u$ exists $\iff$ $m \cdot u$ exists.
Of course two states~$n$ and~$m$ are Nerode-equivalent if and only 
  if~$n\equiv_i m$ for all integers~$i$.

  
  Let us consider two integer~$n$ and~$m$ such that $n \equiv m~[q]$. 
  By induction. 
  The relation~$\sim_0$ is trivial and has only one equivalency class, 
    hence~$n \sim_0 m$.


  Let now be $i$ an integer strictly greater than $0$. 
  Since $n \equiv m~[q]$, it follows from \rlemma{growth} that for every~$n'$
    such that~$n \path{a} n'$, then there exists an integer~$m'$ such 
    that~$m \path{b} m'$ and~$n'\equiv m'~[p]$.
  Hence, from \rdefinition{kr-nai-lab},~${a=(n'\mod{p})=(m'\mod{p})=b}$.
  
  Moreover, by hypothesis~$q\mid p$, hence~$n'\equiv m'~[q]$.
  By induction hypothesis,~$n' \sim_{(i-1)} m'$,
    hence~$n \sim_i m$.
    
    \medskip
The automaton accepting~$\Klgr$ has then~$q+1$ states: one for each congruency 
class modulo~$q$ for positive integers,
plus one special state for~$0$ which is initial. 
\rfigure{k31} shows the case of rhythm~$(3,1)$.
\end{proof}

\begin{figure}[ht!]
  \subfloat[The language~$K_{(3,1)}$]{\lfigure{k31a}%
    \includegraphics[width=.55\linewidth]{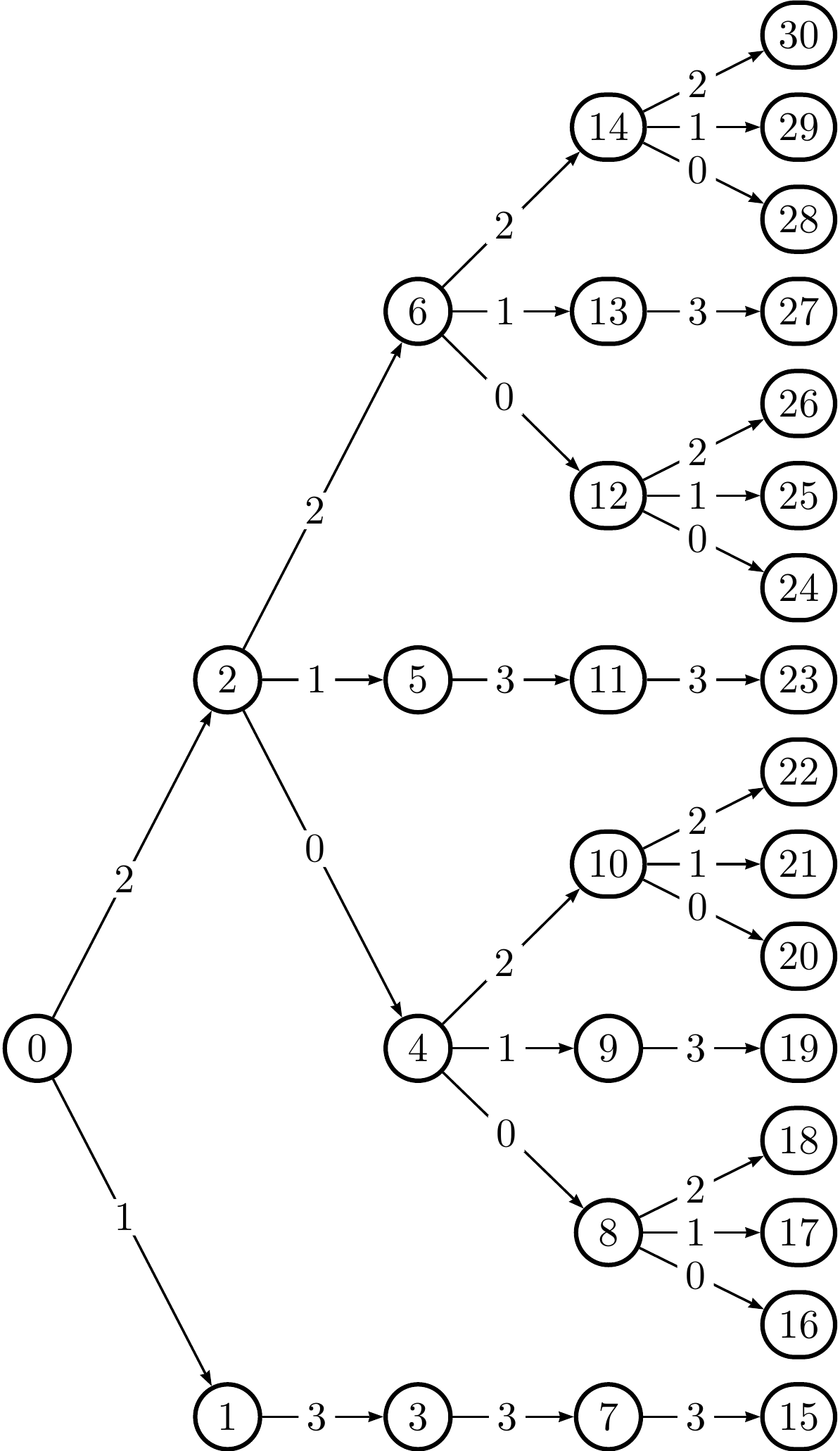}%
  } 
  \hfill
  \subfloat[The automaton accepting~$K_{(3,1)}$]{
  \begin{minipage}[b]{0.35\linewidth}
    \centering
    \includegraphics[width=\linewidth]{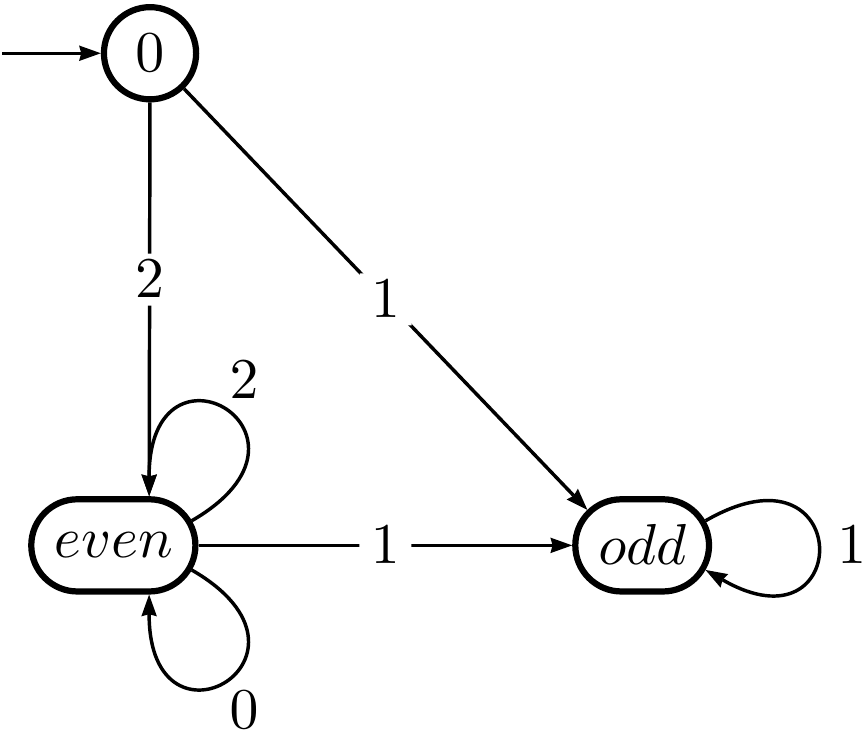}
    \addtocounter{subfigure}{-1}
    \captionof{subfigure}{the automaton accepting~$L_{(3,1)}$}
    \vspace*{2.5cm}
    \includegraphics[width=\linewidth]{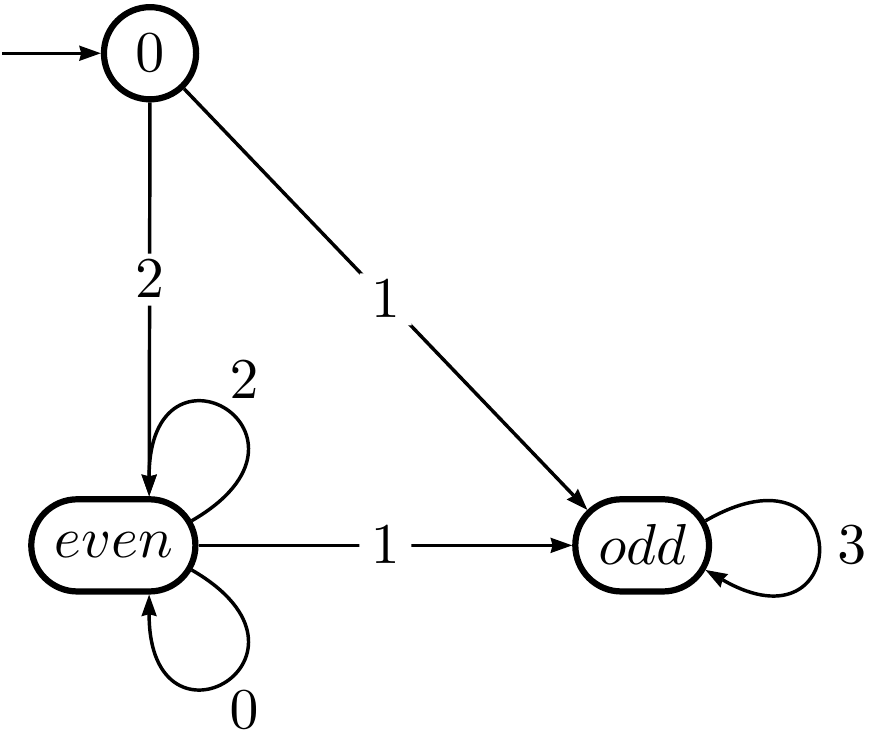}
    \addtocounter{subfigure}{1}
  \end{minipage}
  }
  \caption{The case of the rhythm~$(3,1)$ with integral \gr}
  \lfigure{k31}
\end{figure}

\begin{remark}
If one considers the \itree generated by rhythm~$\rhythm$ (instead of
the tree) then the special case for the state~$0$ of the previous proof 
is unnecessary,
there are only~$q$ states and the congruency class~$(0\mod{q})$ is
initial.
For instance, in \rfigure{k31}c, if there were a self-loop on the
state~$0$, this state would be Nerode-equivalent to the state
\emph{even}.
\end{remark}

More generally, 
if~$\rhth$ is a rhythm of directing parameter~$(q,p)$ 
and~$\labgmm=\labgmtp$ a valid labelling, the language generated 
by~$(\rhth,\labgmm)$ is the branch language of~$\Trr$ labelled by
\begin{equation}
    \forall n,m \in \N \quantsp
	n \pathaut{a}{\Trr} m \e \text{with } a=\gamma_{m\mod{p}} 
	\eqpnt
    \notag
\end{equation}

Of course, the naive labelling can be mapped onto any other (valid) 
labelling:

\begin{proposition}
\lproposition{nlabel>any}%
Let~$L$ be the language of~$B^{*}$ generated by~$(\rhth,\labgmm)$.
There exists a strictly alphabetic 
morphism~$\phi\colon\Aps\rightarrow\Be$ such that~${\phi(\Klgr)=L}$. 
\end{proposition}

We now define a labelling that will play a crucial role in the sequel.

\begin{definition}
\ldefinition{spe-lab}%
Let~$\rhth$ be a rhythm of directing parameter~$(q,p)$
and~$p'$,~$q'$ the coprime integers such that~$\frac{p'}{q'}=\pq$.
Let~$I_{\rhth}$ be the index set of the partial sums of~$\rhth$, 
that is,~${I_{\rhth}=\Defi{r_0+r_1+\cdots+r_k}{\leq k<\qmu}}$.
We call \emph{special labelling} associated with~$\rhth$, and denote 
by~$\labgr$, 
the $p$-tuple ${\labgr=(\gamma_0,\gamma_1,\ldots,\gamma_\ipmu)}$
defined by
\begin{equation}
\gamma_0 = 0 \EqVrgInt
\fa i \not\in I_{\rhth} \quantsmsp 
\gamma_i = \gamma_{(i-1)} +q' \EqVrgInt
\fa i \in I_{\rhth} \quantsmsp 
\gamma_i = \gamma_{(i-1)} +q' -p' 
\eqpnt
\notag 
\end{equation}
\end{definition}

\begin{example}
For instance, if~$\rhth=(3,1,3,3)$, its directing parameter 
is~$(4,10)$, ${p'= 5}$, $q'=2$, $I_{\rhth}=\set{3,4,7}$ and 
\begin{equation}
\labgr = ( \xmd 
    \overbrace{ 0 \xmd , \xmd 2\xmd, \xmd 4 }^{3 }
      \xmd ,\xmd
    \overbrace{1}^{1} 
      \xmd ,\xmd 
    \overbrace{-2\xmd,\xmd0,\xmd 2}^{3}
      \xmd,\xmd 
    \overbrace{-1\xmd,\xmd1\xmd,\xmd3}^3
    \xmd ) 
	\eqpnt
	\notag
\end{equation}
As other examples:
$\msp\labg{(3,1)}=(0,1,2,1)\msp$
and
$\msp\labg{(3,1,1)}=(0,3,6,4,2)\msp$.
\end{example}


It directly follows from the definition that \emph{the special
labelling is valid}.
%


%% file: RBR-sec4.tex

\section{Rational Base Numeration Systems are Rhythmic}
\lsection{rhy-gen}%

In this section, $p$ and~$q$ are two coprime integers, $p>q\jsgeq 1$,
which define the numeration system with base~$\pq$.
We introduce a special, and canonical, \emph{rhythm} of directing
parameter~$(q,p)$,~$\rhthpq$, hence of \gr~$z=\pqs$, which relates to
the classical notion of \emph{Christoffel words}.
We then characterise the \emph{special labelling}~$\labgrpq$ as
the permutation~$\labgpq$ resulting from the generation of~$\Z/p\Z$ 
by~$q$. 
The remarkable fact is then that the representation language in the 
$\pq$-numeration system is generated by the labelled 
rhythm~$(\rhthpq,\labgpq)$ (\theor{lpq-rhythm-perm}). 
  
\subsection{Christoffel words, Christoffel rhythms}
\lsection{Chr-wor-rhy}

\emph{Christoffel words} code some kind of `best approximation' of 
segments the $\Z\x\Z$-lattice and have been studied in 
the field of combinatorics of words (\cf \cite{BersEtAl08}).
We translate them into rhythms.
More precisely (\defin{chr-wor}.a is taken from~\cite{BersEtAl08}):

\begin{definition}
\ldefinition{chr-wor}%
Let~$p$ and~$q$ be two coprime positive integers.

\tha 
The (upper) \emph{Christoffel word} 
  associated with \pq*, and denoted by~$\wpq$, is the label of the path 
  from $(0,0)$ to~$(q,p)$ on the~$\Z\x\Z$ lattice, 
  such that 
  \begin{itemize}
    \item the path is above the line of slope~\pq* passing through the origin;
    \item the region  enclosed by the path and the line 
          contains no point of~$\Z\x\Z$.
  \end{itemize} 

\thb The \emph{Christoffel rhythm} associated with~$\pqs$, and denoted
by~$\rpq$, is the rhythm whose path is~$\wpq$:
$\msp\pthw{\rpq}=\wpq\msp$, 
hence its directing parameter is~$(q,p)$. 
\end{definition}

\figur{christ-221} shows the path 
of~$\msp\wct = \OBK{y\xmd y}\xmd 
        x\xmd\OBK{y\xmd y}\xmd 
		x\xmd\OBK{y}\xmd x\msp$, 
the Christoffel	word	associated with~$\cts$;
then,~$\rhthct=(2,2,1)$.
Other Christoffel words and their paths are shown at 
\figur{a-christoffel}.

  
\begin{figure}[ht!]     
\renewcommand{\scaleu}{1.2}%
\e
\subfloat[$\frac{3}{2}:~yy\,x\,y\,x$]{\includegraphics[scale=\scaleu]{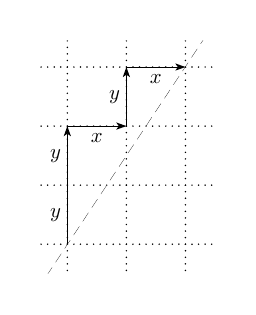}}
\hfill
\subfloat[$\frac{5}{2}:~yyy\,x\,yy\,x$]{\includegraphics[scale=\scaleu]{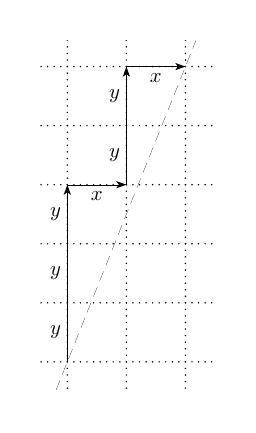}}
\hfill
\subfloat[$\frac{5}{3}:~yy\,x\,yy\,x\,y\,x$]{\includegraphics[scale=\scaleu]{christ-ct}}
\e
\caption{Christoffel words associated with three rational numbers}
\lfigure{a-christoffel}
\renewcommand{\scaleu}{1.4}%
        \medskipneg
\end{figure}

\begin{proposition}
        \lproposition{ceil}%
  Given a base~\pq* of rhythm~$\rpq=(r_0,r_1,\cdots,r_{\iqmu})$, 
  and an integer~${0 < k \leq q}$,
    the partial sum~$r_0+r_1+\cdots+r_{k-1}$, of the first $k$ components 
    of~$\rhythm$ is equal to the smallest integer
    greater than~$k\pq$.
\end{proposition}
\lproposition{cword->chry}
                          %
\propo{ceil}
is a direct consequence of the following technical lemma
which is basically a translation of the proposition into the
Christoffel words and their geometric interpretation universe.
  
\begin{lemma}
\llemma{begrhy}%
Let us denote by~$\wpq$ the Christoffel word of slope~$\pqs$.
If~$u\xmd x$ is prefix of~$\wpq$ then it corresponds to
  a path from~$(0,0)$ to~$(k,\ceil{k\pq})$ in the~$\Z\x\Z$ lattice, 
  where~$k$ is the number of~$x$'s in~$u\xmd x$.
\end{lemma}

\begin{proof} 
  From \rdefinition{chr-wor} of Christoffel word, there is no integer point
    between the path and the line of slope \pq* and passing through the origin.
    
  Since the point~$(k,k\pq)$ is part of this line, the Christoffel path must 
    pass through the point~$(k,\ceil{k\pq})$.
  Besides, the prefix of the Christoffel word reaching this point must end with 
    an~$x$; were it ending with an~$y$ it would mean that the Christoffel path
    pass through the point~$(k,\ceil{k\pq}-1)$ which is 
    \emph{below} the line of slope \pq*, a contradiction.
\end{proof}

\begin{lemma}
\llemma{endrhy}%
Given a base \pq*, we denote the associated Christoffel rhythm
by~$\rhythm$ and an integer~$k\in\intint{0}{\qmu}$,
$\sum_{i=0}^{\ikmu} r_{(q-1-k)} = \floor{k \xmd \pq}$
\end{lemma}

We define the sequence of integers~$e_0,e_1,\ldots,e_{\qmu}$ such 
  that~$e_{j}$ is the difference between
  the approximation~$(r_0+r_1+\cdots+r_{\ikmu})$ and the point of the 
  associated line of the respective abscissa, that is~$(k\xmd\pq)$.
This difference is a rational number smaller than~$1$ and whose 
  denominator is~$q$, in order 
  to obtain an integer we multiply it by~$q$:
  
\begin{equation}\label{eq.gap-def}
  \forall k\in \intint{0}{\qmu} \quantsp
      e_k \e = \e q\left(\xmd\sum_{i=0}^{k-1} r_i \right)- k \xmd p 
          \eqpnt
\end{equation}

We describe on the example of the base~$\frac{5}{3}$ shown at 
\rfigure{rhythm-ct} a more diagrammatic way of characterising 
Christoffel rhythms.
We associate~$p$ segments of length~$q$ (at the top in the figure)
with~$q$ segments of length~$p$ such that each top segment is associated
to the bottom segment in which it starts.
The integer~$e_i$ is then the difference of length between the~$i$
left-most bottom segments (of length~$p \times i $) and the total
number of top segments associated with them (of length~${q\times
\sum_{j=0}^{\iimu} r_j}$).

\begin{figure}[h!]
  \centering
  \includegraphics[width=\linewidth]{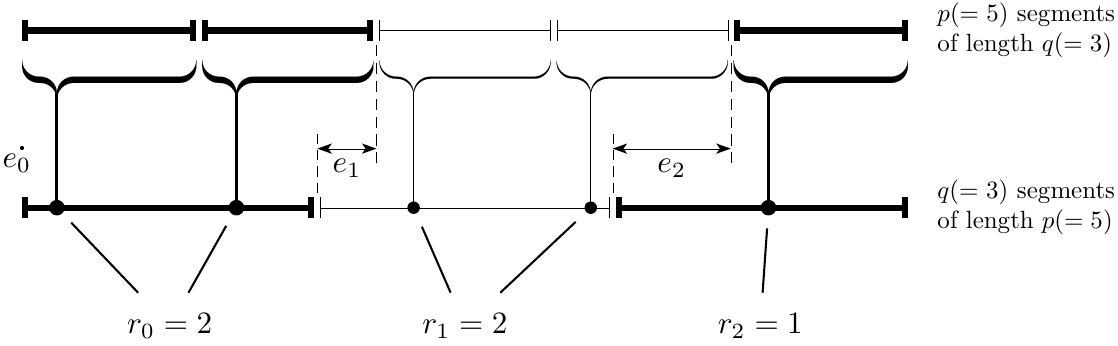}
  \caption{Diagrammatic interpretation of the rhythm~$(2,2,1)$ of base~$\frac{5}{3}$}
  \lfigure{rhythm-ct}
\end{figure} 

The next lemma compiles basic properties of the~$r_j$'s and~$e_j$'s.

\begin{lemma}
\llemma{rhythm-prop}%
Let $\rpq=(r_0,r_1,\ldots,r_{\qmu})$ be the {Christoffel rhythm} 
of slope~$\pq$.
For every integer~$j$ in~$\intint{0}{q-1}$, it holds:
\begin{enumerate}%
        [label=\textnormal{(\alph{*})}\,\,,leftmargin=2.5em,%
         topsep=1ex,itemsep=1ex,%
         ref=\protect\rlemma*{rhythm-prop}\alph{*}]

\item \llemma{ej-Aq}
      $e_j$ belongs to $\intint{0}{\qmu}$; 

\item \llemma{smallest}
      $r_j$ is the smallest integer such that~$\msp q \xmd r_j + e_j \geq p\msp$;

\item \llemma{rj(ej)}
      $\msp e_{j+1} = e_j + q\xmd r_j -p $ ;
    
\item \llemma{ej-carac} all~$e_j$'s are distinct;
    
\item \llemma{midrhy2}
      $\msp\displaystyle{q \sum_{k=i}^{\ijmu} r_k +e_i
                  = (j-i)\xmd p + e_{j}}$ ;
    
\item \llemma{midrhy}
for every~$i$ in~$\intint{1}{q}$,~$i>j$, it holds:
\begin{equation}
\sum_{k=i}^{j-1}r_k = \bceil{\frac{(j-i)\xmd p\xmd -\xmd e_i}{q}} 
                    =\bfloor{\frac{(j-i)\xmd p\xmd -\xmd e_j}{q}}
\eqpnt
\notag
\end{equation}
  \end{enumerate}
\end{lemma}

\begin{proof}
        
Statements (a), (b), and (c) are simple consequences of the
definition.

\smallskip

(d)\e 
Let~$i$ and~$j$ be in~$\intint{0}{\qmu}$ and suppose that~$e_i=e_j$. 
It follows from~(b) and~(c) that~$r_i = r_j$ and then~$e_{i+1}=e_{j+1}$. 
By iterating this process, it follows that the rhythm~$\rhythm$ is
periodic of period~$\wlen{j-i}$, a contradiction.
 
\smallskip

(e)\e 
Follows from the iteration of~(c).
  
\smallskip
  
(f)\e  
From (e),
$\msp \displaystyle{e_{i+j} = 
   e_i + q\xmd \left(\sum_{k=i}^{j-1} r_k\right) - (j-i)\xmd p}\msp$,
hence
\begin{equation}
\sum_{k=i}^{\ijmu} r_k = \frac{(j-i)\xmd p}{q}-\frac{e_i}{q} + \frac{e_{j}}{q}
\eqpnt
\notag
\end{equation}
Since the right-hand side of this equation is an integer and that~$e_i$ 
and~$e_j$ are smaller than~$q$, $\frac{e_i}{q}$ and $\frac{e_j}{q}$ 
are smaller than~$1$ and the whole statement follows.
\end{proof}

\subsection{Generation of~$\Lpq$ by rhythm and labelling}
\lsection{gen-lpq-rhy}

Since~$p$ and~$q$ are coprime integers,~$q$ is a generator of the 
(additive) group~$\Z/p\Z$.
We denote by~$\labgpq$ the sequence induced by this generation 
process: 
\begin{equation}
\labgpq = 
(\,0,\,q\mod{p},\,(2\xmd q)\mod{p},\ldots,\,((\pmu)\xmd q)\mod{p}\,)
    \notag
	\eqpnt
\end{equation}

\begin{proposition}
\lproposition{seq-gen-spe}%
Let $p$ and~$q$ be two coprime positive integers, $p>q\jsgeq 1$.
The sequence~$\labgpq$ coincide with the special labelling~$\labgrpq$. 
\end{proposition}
\begin{proof}
  We denote by~$(\gamma_0,\gamma_1,\ldots,\gamma_{\iqmu})$ the special 
    labelling~$\labelling[\gamma]_{\rpq}$.
  By the very definition of special labelling, 
    it is obvious that for 
    all~$i\in\intint{0}{\qmu}$,~$\gamma_i\equiv~i\xmd q~[p]$.
  It is then enough to prove that for every 
    integer~$i\in\intint{0}{\qmu}$,~$0\leq\gamma_i<p$.
  Given an integer~$k$ we denote by~$m_k = r_0+r_1+\cdots+r_{k-1}$
    and by~$M_k=(m_{k+1}-1) = r_0+r_1+\cdots+r_{k-1}+ r_k - 1$.
  
  Let us fix an integer~$i$ in the following; 
  there exists an integer~$k$ such that~$m_k \leq i \leq M_k$.
  From the definition of special labelling,~${\gamma_{m_k} = (q\xmd m_k) -(p\xmd k)}$,
      ~${\gamma_{i} = (q\xmd i) -(p\xmd k)}$ 
      and~${\gamma_{M_k} = (q\xmd M_k) - (p \xmd k)}$,
      hence~$\gamma_{m_k} \leq \gamma_i \leq \gamma_{M_k}$.

  From \rproposition{ceil},~$m_k$ is the smallest integer greater than~$k\frac{p}{q}$, 
    hence~${m_k\geq k\pq}$, hence~$\gamma_{m_k} \geq0$.
  Using the same proposition for~$(k+1)$,~$m_{k+1}$ is the smallest integer greater 
    than~$(k+1)\pq$, hence~$M_k (=m_{k+1} -1)$ is strictly 
    smaller than~${(k+1)\pq}$; that is,~$M_k<(k+1)\pq$, hence~$\gamma_{M_k} < p$.
  
  Therefore~$0 \leq \gamma_{m_k} \leq \gamma_i \leq \gamma_{M_k} < p$. 
  
\end{proof}

As an immediate consequence of \propo{seq-gen-spe}, $\labgpq$ is a 
\emph{valid labelling}.

\begin{theorem}
\ltheorem{lpq-rhythm-perm}%
Let $p$ and~$q$ be two coprime integers, $p>q\jsgeq 1$.
The language \Lpq* of $\pq$-representations of the integers
is generated by the rhythm~$\rpq$ and labelling~$\prmpq$.
 \end{theorem}

For instance,~$L_\frac{3}{2}$, shown at \figur{l32}, is built
  with the rhythm~$(2,1)$ and the labelling~$(0,2,1)$.
  The proof of \rtheorem{lpq-rhythm-perm} relies mostly on the following statement
  itself being a consequence of the technical \rlemma{smallest}.
  
\begin{proposition}\lproposition{rhythm}
  For every integer $n>0$ (resp.~${n=0}$), there is exacly~$r_{(n\imod{q})}$ 
    (resp.~$(r_0-1)$) letters a of \Ap* such that~$\cod{n}.a$ is in \Lpq*.
\end{proposition}

\begin{proof}
  We denote by~$j$ the congruency class of~$n$ modulo~$q$.
  From \rlemma{smallest},~$r_j$ is the smallest integer
    such that~$q\xmd r_j + e_j > p$.
  It follows that for all~$k$ in~${\intint{0}{r_{j}-1}}$ ~${(e_j+q\xmd k) < p}$ 
    and~${e_j +q\xmd r_j > p}$.
    
  From \rlemma{smallest}, $e_j$ is the smallest label of the 
    state~$n$, hence the state~$n$ has exactly~$r_j$ outgoing transitions, 
    respectively labelled by~${e_j,e_j+q,\ldots, (e_j+q\xmd(r_j-1)})$. 
\end{proof}

\begin{lemma}\llemma{ej-smallabel}
  Given a base \pq* and an integer~$n$, the smallest letter~$a$ of \Ap* such that 
    $\cod{n}.a$ is in \Lpq*, is~$e_{(n\mod{q})}$.
\end{lemma} 

\begin{proof} 
  Let us denote by~$n$ an integer and by~$j$ its congruency modulo~$q$.
Since $e_{j}$ is in \Aq* (from \rlemma{ej-Aq}), it is enough to prove 
    that~$e_{j}$ is an outgoing label of~$n$, or (from \requation{pathlpq})
    that~$n\xmd p + e_{j}$ is a multiple of~$q$; or, equivalently
    that~$j\xmd p + e_{j}$ is a multiple of~$q$.
From \requation{gap-def},~$j\xmd p + e_j = \left(\xmd q\xmd \sum_{i=0}^{\ijmu} r_i \xmd \right)$,
  that is, a multiple of~$q$.
\end{proof}

To complete the proof of \rtheorem{lpq-rhythm-perm}, it remains to prove that
for every integer~$n$, the last digit of~$\pqRep{n}$ is~$(q\xmd n) \mod{p}$,
which directly results from the definition of the modified Euclidean 
division algorithm (\requation{mod-euc-div}).
%
%
%
%
%


%% file: RBR-sec5.tex
\section{Reduction to Rational Base Numeration Systems}
\lsection{red-to-rat}%

In this section, $p$ and~$q$ are two integers, $p>q\jsgeq 1$, not 
necessarily coprime, and~$\rhth$ is a rhythm of directing 
parameter~$(q,p)$.
As in \defin{rhythm-growth}, we denote by~$p'$ and~$q'$ their 
respective quotient by their gcd, that is, 
$\pqps$ is the reduced fraction of~$\pqs$.

\begin{definition}\ldefinition{lr}
We denote by~$\Lr$ the language generated by a rhythm~$\rhth$
and the associated special labelling~$\labgr$.
\end{definition}

If~$\rhth$ happens to be a Christoffel rhythm, 
then~$\Lr$ is by definition equal to~$\Lpq$ (which, in this case, is 
the same as~$\Lpqp$).
The key result of this work states that~$\Lr$ is indeed of the same 
kind as~$\Lpqp$.

\begin{theorem}
\ltheorem{lr-num-pos}%
Let~$\rhth$ be a rhythm of directing parameter~$(q,p)$ and~$\pqps$
the reduced fraction of~$\pqs$.
Then, the language~$\Lr$ is a set of representations of the integers 
in the rational base~$\pqps$.
\end{theorem}
Even though~$p$ and~$q$ are not coprime, the arcs of the 
tree \Lr* satisfies essentially the same equation as~$\Lpqp$ 
(\cf \requation{pathlpq}) as expressed by the following statement.
 
\begin{lemma}
\llemma{a-lr-num-pos}%
Let~$\rhth$ be a rhythm of directing parameter~$(q,p)$ and~$\pqps$
the reduced fraction of~$\pqs$.
Then, for every integers~$n$ and~$m$ it holds: 
\begin{equation}
n\pathaut{a}{\Lr}m
\e\Longrightarrow\e
a = q'\xmd m - p'\xmd n
\eqpnt
\notag
\end{equation}
\end{lemma}

\begin{proof}
By induction on~$m$. 
The implication obviously holds for the first arc of the tree~$\Lr$ 
as it is~$0\pathaut{q'}{\Lr}1$.

Let us assume it holds for the~$m$-th arc, that is,
$\msp n\pathaut{a}{\Lr}m\msp$ with~$a = q' m -p'n$.
The~$(m+1)$-th arc is either 
$\msp n\pathaut{b}{\Lr} (m+1)\msp$
or
$\msp(n+1) \pathaut{b}{\Lr} (m+1)\msp$.

\begin{itemize}
    
\item[$\bullet$] $\msp n\pathaut{b}{\Lr} (m+1)\msp$ 
corresponds to the case where 
$\msp\displaystyle{(m+1)\mod{p} < \sum_{i=0}^{n\imod q} r_i}\msp$, 
hence
\begin{equation}
b = \gamma_{((m+1)\imod{p})} = \gamma_{(m\imod{p})} + q' 
                             = q' (m+1) - p'n
\eqpnt
        \notag
\end{equation}
        
\item[$\bullet$] $\msp(n+1) \pathaut{b}{\Lr} (m+1)\msp$
corresponds to the case where
$\msp\displaystyle{(m+1)\mod{p} = \sum_{i=0}^{n\imod q} r_i}\msp$, 
hence
\begin{equation}
b = \gamma_{((m+1)\imod{p})} = \gamma_{(m\imod{p})} + q'-p' 
                             = q'(m+1) - p'(n+1)
\eqpnt 
        \notag
\end{equation}
\end{itemize}
In both cases, the second equality follows from the definition of the 
special labelling~$\labgr$.
\end{proof}

If we call $\rhth$-representation of an integer~$n$, and denote 
by~$\rRep{n}$, the word that labels the path from the root~$0$ to the 
node~$n$ in the labelled tree defined by~$\Lr$, \theor{lr-num-pos} is 
equivalent to the following statement that is established by induction 
on the length of the $\rhth$-representation of~$n$.

\begin{proposition}
\lproposition{lr-num-pos}%
Let~$\rhth$ be a rhythm of directing parameter~$(q,p)$,~$\pqps$
the reduced fraction of~$\pqs$ and~$\fvalpqp$ the evaluation 
function in the numeration system with rational base~$\pqps$.
Then, for every integer~$n$ it holds: 
\begin{equation}
\lequation{lr-num-pos}
\evalpqp{\rRep{n}}=n
\eqpnt
\notag
\end{equation}
\end{proposition}
\begin{proof}
Let
$\msp \rRep{n} = a_k a_{k-1}\cdots\xmd a_0\msp$ be the 
$\rhth$-representation of~$n$. 
We then want to prove that
\begin{equation}
n = \sum_{i=0}^{k} \frac{a_i}{q'} \xmd \left(\pqps\right)^i
\eqpnt
\notag
\end{equation}

By induction on the length of~$\rRep{n}$. 
The equality obviously holds true for~$\rRep{0} = \epsilon$.
  
Let~$m$ be an integer and~$\rRep{m}=a_{k+1}\xmd a_{k}\xmd a_{k-1}\cdots\xmd a_1\xmd a_0$ 
its $\rhth$-representation, that is, a word of~$\Lr$.
The word~$a_{k+1}\xmd a_k\xmd a_{k-1}\cdots\xmd a_1$ is also 
in~$\Lr$; it is the $\rhth$-representation of an integer~$n$ strictly smaller than~$m$,
and such that:
\begin{equation}
n\pathaut{a_0}{\Lr}m
\eqpnt
\notag
\end{equation}
By induction hypothesis, 
$\msp \displaystyle{n=\sum_{i=1}^{k+1} 
           \frac{a_i}{q}\xmd\left(\pqps\right)^{\iimu}}\msp$.
It follows from \lemme{a-lr-num-pos} 
that~$\msp a_0  = q'm - p'n\msp$, or, equivalently, 
that~$\msp\displaystyle{m=\frac{np'+a_0}{q'}}\msp$, 
hence
\begin{equation}
m \e=\e \pqps \left(\sum_{i=1}^{k+1} \frac{a_i}{q'} \xmd 
                               \left(\pqps\right)^{i-1}\right)
        + \frac{a_0}{q}
  \e=\e  \left(\sum_{i=1}^{k+1} \frac{a_i}{q'} \xmd 
                               \left(\pqps\right)^{i}\right)
        + \frac{a_0}{q}
\eqpnt
\notag
\end{equation}
\end{proof}

In other words, $\Lr$ seen as an abstract numeration system is indeed 
a \emph{positional} numeration system.

\medskip 

It has been shown in~\cite{AkiyEtAl08} that every numeration system 
in rational base~$\pqs$ has the remarkable property that even though 
the representation language~$\Lpq$ is not a regular language, the 
\emph{conversion} from any digit-alphabet~$B$ into the canonical 
alphabet~$\Ap$ is realised by a \emph{finite transducer} (indeed a 
letter-to-letter right sequential transducer), exactly as in the case 
of the numeration system in base~$p$ (\cf also~\cite{FrouSaka10hb}).

More precisely, let~$\Br$ be the digit-alphabet of the special 
labelling~$\labgr$.
Let~$\chi_{\rhth}$ be the function from~$\Brs$ into~$\Apps$ which 
maps every word of~$\Brs$ onto the word of~$\Apps$ which has the same 
value in the numeration system in the base~$\pqps$, that is, 
\begin{equation}
	\fa w \in\Brs \quantsp
\evalpqp{w} = \evalpqp{\chi_{\rhth}(w)}
\eqpnt
\notag
\end{equation}
Hence
$\msp\chi_{\rhth}(\Lr)=\Lpqp\msp$.
And we can then state:

\begin{theorem}[\cite{AkiyEtAl08}]
	\ltheorem{converter}%
The map~$\chi_{\rhth}$ is a rational function.
\end{theorem}

We can now complete the proof of \theor{trans-language}.b.
Since~$\Klgr$ is prefix-closed, it is a FLIP language if and only if 
it contains no infinite regular subset.
Suppose that~$\Klgr$ contains an infinite regular subset~$R$.
There exists a morphism~$\phi_{\rhth}$ 
such that~$\phi_{\rhth}(\Klgr)=\Lr$ 
and a rational function~$\chi_{\rhth}$
such that~$\chi_{\rhth}(\Lr)=\Lpqp$;
hence the FLIP language~$\Lpqp$ contains the infinite regular 
subset~$\chi_{\rhth}(\phi_{\rhth}(R))$, a contradiction.
  

%% file: RBR-ccl.tex
\section{Conclusion and future work}

With this notion of labelled signature, we have somehow captured the 
`regularity' of the representation languages in rational bases by 
showing that they have \emph{periodic} labelled signatures and that 
this periodicity is to some extent characteristic of these languages.
A by-product of this characterisation is the remarkable fact that 
periodic labelled signatures yield either very simple languages (when 
the growth ration is an integer) or vey complex (when it is not).

It would be very tempting to get the same kind of results with 
periodic signatures only, that is, without bringing labelling into 
play.
On one hand, it is easy to show that one gets a rational tree (that 
is, a tree with finite distinct subtrees) in the case of an 
integral growth ratio.
But on the other hand, and as it is related to open problems in number 
theory,  it would be certainly difficult to show, for instance, that 
all subtrees are distinct in the case of a non-integral growth ratio, 
although it is a reasonable conjecture.
Hopefully, there are easier problems at hand.

For sake of simplicity, we have considered here purely periodic 
signatures and labelled signatures only.
The generalisation to ultimately periodic ones raises no special 
difficulties but technical details to be settled.
And the results established here readily extend.

The problem of the representation of negative integers in the 
$\pqs$-numeration systems was considered (among others) 
in~\cite{FrouKlou12a}.
The characterisation of these $\pqs$-numeration systems by the 
corresponding Christoffel words and the study of their combinatorial 
properties allow a new approach to this problem and yield new proof 
to some results of~\cite{FrouKlou12a}; it is the purpose of 
forthcoming work of the first author~\cite{Marsxx}.


%% file: RBR-bib.tex
\bibliographystyle{plain}
\small
\bibliography{%
  ../bibliography,%
  Alexandrie-abbrevs,%
  Alexandrie-AC,%
  Alexandrie-DF,%
  Alexandrie-GL,%
  Alexandrie-MR,%
  Alexandrie-SZ%
}
\normalfont
 